\documentclass[12pt,letterpaper]{JHEP3}

\usepackage{amscd,amsmath,amssymb,amsfonts,xspace,mathrsfs,amsthm}
\usepackage{color}
\let\normalcolor\relax
\usepackage{bbold}
\usepackage{latexsym}
\usepackage{graphicx}
\usepackage{dsfont}
\usepackage{color}
\usepackage{longtable}
\usepackage[bbgreekl]{mathbbol}
\usepackage{bbm}
\usepackage{upgreek}

\DeclareMathAlphabet{\euscr}{U}{eus}{m}{n}

\def\bea{\begin{eqnarray}}
	\def\eea{\end{eqnarray}}
\def\be{\begin{equation}}
	\def\ee{\end{equation}}
\def\ba{\begin{align}}
	\def\ea{\end{align}}
\def\bse{\begin{subequations}}
	\def\ese{\end{subequations}}

 \newtheorem{theorem}{Theorem}
  \newtheorem{proposition}{Proposition}
  \newtheorem{lemma}{Lemma}
  
\def\det{\,{\rm det}\, }

\def\sign{{\rm sgn}}

\def\Ch{{\rm Ch}}

\def\Sym{\,{\rm Sym}\, }
\def\cSym{\,{\rm cSym}\, }

\def\Span{\,{\rm Span}\, }

\DeclareMathOperator{\Erf}{Erf}
\DeclareMathOperator{\Erfc}{Erfc}

\newcommand{\sgn}{{\rm sgn}}
\newcommand{\rank}{\mbox{rank}}

\def\({\left(}
\def\){\right)}
\def\[{\left[}
\def\]{\right]}
\def\<{\left\langle}
\def\>{\right\rangle}
\def\hf{{1\over 2}}

\newcommand{\p}{\partial}

\newcommand{\vth}{\vartheta}

\def\vrh{\varrho}

\newcommand{\eps}{\epsilon}
\newcommand{\veps}{\varepsilon}

\newcommand{\de}{\mathrm{d}}

\newcommand{\I}{\mathrm{i}}

\newcommand{\rmR}{\mathrm{R}}

\newcommand{\cL}{\mathcal{L}}
\newcommand{\cD}{\mathcal{D}}
\newcommand{\cF}{\mathcal{F}}

\newcommand{\cG}{\mathcal{G}}
\newcommand{\cB}{\mathcal{B}}

\newcommand{\cM}{\mathcal{M}}

\newcommand{\cJ}{\mathcal{J}}

\newcommand{\cI}{\mathcal{I}}
\newcommand{\cO}{\mathcal{O}}

\newcommand{\xbbm}{\mathbbm{x}}

\newcommand{\vbbm}{\mathbbm{v}}

\newcommand{\bbbm}{\mathbbm{b}}
\newcommand{\ebbm}{\mathbbm{e}}

\newcommand{\bsf}{{\sf b}}

\newcommand{\qsf}{{\sf q}}

\newcommand{\Msf}{\mathsf{M}}

\newcommand{\Ssf}{\mathsf{S}}
\newcommand{\esf}{\mathsf{e}}

\newcommand{\IT}{\mathds{T}}
\newcommand{\IR}{\mathds{R}}
\newcommand{\IC}{\mathds{C}}
\newcommand{\IZ}{\mathds{Z}}

\newcommand{\IN}{\mathds{N}}

\newcommand{\IF}{\mathds{F}}

\newcommand{\IP}{\mathds{P}}

\def\scR{\mathscr{R}}
\def\Mv{\mathscr{M}}

\def\Bv{\mathscr{B}}
\def\Ev{\mathscr{E}}
\def\Fv{\mathscr{F}}

\def\Zv{\mathscr{Z}}

\def\euM{\euscr{M}}
\def\euF{\euscr{F}}

\newcommand{\bffv}{{\bf v}}

\newcommand{\bfv}{{\boldsymbol v}}

\newcommand{\bfp}{{\boldsymbol p}}
\newcommand{\bfq}{{\boldsymbol q}}

\newcommand{\bfx}{{\boldsymbol x}}
\newcommand{\bfy}{{\boldsymbol y}}

\newcommand{\bfN}{{\boldsymbol N}}

\newcommand{\bftet}{{\boldsymbol \theta}}
\newcommand{\bflam}{{\boldsymbol \lambda}}
\newcommand{\bfmu}{{\boldsymbol \mu}}

\newcommand{\bfLam}{{\boldsymbol \Lambda}}

\newcommand{\bfgam}{{\boldsymbol \gamma}}

\def\tV{\tilde V}

\def\tcB{\tilde\cB}
\def\tcM{\tilde\cM}

\def\tvbbm{\tilde\vbbm}

\def\ba{\bar a}

\def\btau{\bar \tau}

\def\bOm{\overline\Omega}

\def\hq{\hat q}

\def\hM{\hat M}

\def\hPhi{\hat\Phi}

\def\hgam{\hat\gamma}
\def\bfhgam{\hat\bfgam}

\def\hMsf{\hat\Msf}

\def\CY{\mathfrak{Y}}

\def\vert{v}
\def\vth{\vartheta}

\def\hint{h^{\rm (int)}}

\def\Ef{\Ev^{(0)}}
\def\Efrf{\Ev^{(0)}}

\def\Er{\Ev}

\def\Ep{\Ev^{(+)}}
\def\Eprf{\Ev^{(+)}}

\def\rmRirf#1{\rmR^{(#1)\rm ref}}

\def\scRrf{\scR^{\rm ref}}

\def\hr{h^{\rm ref}}

\def\whhr{\widehat h^{\rm ref}}

\def\whh{\widehat h}

\def\whTheta{\widehat\Theta}

\def\wheuF{\widehat \euF}

\def\whFv{\widehat\Fv}

\def\Fp{\cF^{\perp}}

\def\vu{\mathfrak{u}}

\def\vb{\mathfrak{b}}

\def\vm{\mathfrak{m}}

\def\tvm{\tilde\vm}
\def\tvu{\tilde\vu}

\def\tvb{\tilde\vb}

\def\bbp{\mathbb{p}}
\def\bbv{\mathbb{v}}

\newcommand{\q}{\mathrm{q}}

\def\glu{{\rm g}}

\def\nv{v_0}

\def\mI{m_\cI}

\def\hbbgam{\hat\bbgamma}

\title{Remarks on Vafa-Witten theory, \\ or how to sit on a wall
	}

\author{Sergei Alexandrov$^1$ and Aradhita Chattopadhyaya$^2$ 
\\
$^1$ {\it Laboratoire Charles Coulomb (L2C), Universit\'e de Montpellier,
	CNRS, F-34095, Montpellier, France}\\
$^2$ {\it Chennai Mathematical Institute, H1 SIPCOT IT PARK, Siruseri, Kelambakkam, TN 603103, India}
\\
\vspace*{2mm} {\tt e-mail:
	\email{sergey.alexandrov@umontpellier.fr}, \email{aradhitac@cmi.ac.in}
}

\vspace*{-3mm}

}

\abstract{We prove three results about refined $SU(N)$ VW invariants on Hirzebruch surfaces. 
	First, we crucially simplify the modular anomaly equation satisfied by the generating series of these invariants 
	by expressing it in terms of complementary generalized error functions,
	which requires extending these error functions to their singular loci.
	Second, we evaluate the variation of the modular completions of the generating series upon change of polarization.
	Finally, a representation of the generating functions in terms of indefinite theta series 
	allows us to compute VW invariants at the walls of marginal stability,
	and to express them in terms of the invariants on the two sides of the wall and those of lower ranks.
	All the results exhibit a universal structure suggesting their wider applicability.
}

\begin{document}

\section{Introduction}
\label{sec-intro}

Vafa-Witten (VW) theory is a topological field theory characterized by a gauge group 
and defined on any complex surface $S$ \cite{Vafa:1994tf}. It can be solved by localization 
and in favorable circumstances its partition function is expressed in terms of topological invariants of the surface,
known as {\it VW invariants}, which count the Euler numbers of the moduli spaces of instantons on $S$.
Upon introducing {\it refinement}, parametrized by a complex variable $y$, a fugacity conjugate to the angular momentum, 
one gets access to the Betti numbers of these moduli spaces, called {\it refined VW invariants}.

A remarkable fact is that VW invariants manifest {\it modularity}, which appears as a consequence of S-duality 
of the underlying physical theory \cite{Vafa:1994tf}. 
In particular, it implies that the partition function must be a modular form, 
or a Jacobi form in the refined case.
Furthermore, for surfaces with $b_2^+(S)=1$, 
VW invariants, similarly to other topological invariants such as 
Seiberg–Witten or Donaldson invariants (see, e.g., \cite{Manschot:2023rdh}),
experience wall-crossing and their generating functions acquire a {\it modular anomaly}.\footnote{For results about surfaces with $b_2^+(S)>1$ see, e.g., 
	\cite{Vafa:1994tf,Labastida:1999ij,gottsche2017virtual,Laarakker:2018isn,Thomas:2018lvm,Gottsche:2018meg,Gottsche:2021dye}.}
However, the anomaly makes the story even more beautiful because it appears to be a manifestation of {\it mock modularity} 
\cite{Zwegers-thesis,MR2605321,Alexandrov:2025sig}.
For instance, for the gauge group $G=SU(N)$, the generating functions turn out to be (vector valued) depth $N-1$ mock modular forms, 
or mock Jacobi forms in the presence of refinement \cite{Vafa:1994tf,Minahan:1998vr,Manschot:2017xcr,Alexandrov:2019rth}.
The partition function is then given by a non-holomorphic {\it modular completion} of 
the generating series and hence exhibits a {\it holomorphic anomaly}.
As was shown explicitly in \cite{Dabholkar:2020fde} for the simplest case of $SO(3)$ VW theory on $\IP^2$, 
the non-holomorphic anomalous contributions to the path integral 
originate from the $Q$-exact terms due to boundaries of the moduli space, 
in parallel to the holomorphic anomaly in topological string theory \cite{Bershadsky:1993ta}.

The precise form of the modular anomaly has been derived in \cite{Alexandrov:2019rth}, starting from a similar anomaly
for D4-D2-D0 BPS indices, which are physical realizations of rank 0 generalized Donaldson-Thomas (DT) invariants
of Calabi-Yau (CY) threefolds (see \cite{Alexandrov:2025sig} for a review or \cite{Alexandrov:2012au,Alexandrov:2016tnf,Alexandrov:2018lgp}
for original works). The point is that $SU(N)$ VW invariants of a surface $S$ coincide with these BPS indices provided the CY is taken to be 
the canonical bundle ${\rm Tot}(K_S)$ over this surface 
\cite{Minahan:1998vr,Alim:2010cf,Gholampour:2013jfa,gholampour2017localized},
while the rank $N$ is identified with the number of D4-branes wrapped on $S$.
Furthermore, the non-compactness of the CY allows us to introduce refinement, and the identification 
between the BPS indices and the VW invariants continues to hold also at the refined level.

The resulting anomaly is conveniently formulated as an equation relating the modular completion 
$\whh_{N,\mu}$ of the generating function $h_{N,\mu}$ to the generating functions $h_{N_i,\mu_i}$ of lower ranks $N_i<N$
(see \eqref{exp-whhr}).
Combined with some additional restrictions (well-defined unrefined limit, match to known results for small $N$, etc.),
this equation can actually be solved for the generating functions and their completions.
This idea has been realized for $\IP^2$, Hirzebruch and del Pezzo surfaces in \cite{Alexandrov:2020bwg,Alexandrov:2020dyy},
generalizing previous results of 
\cite{Yoshioka:1994,Yoshioka:1995,Bringmann:2010sd,Manschot:2010nc,Manschot:2011dj,Manschot:2011ym,Klemm:2012sx,
	Manschot:2014cca,Manschot:2017xcr,Beaujard:2020sgs}.

In fact, the anomaly equation restricts only the VW invariants evaluated in the {\it canonical chamber}
in the moduli space parametrized by polarization vector $J$, which is defined as the chamber containing the point $J=c_1(S)$.
However, the work \cite{Alexandrov:2020dyy} suggested a natural extension of the above solution to arbitrary $J\in\mbox{\rm Span}(c_1(S),\nv(S))^+$.
The only remaining restriction is that it belongs to the projection of the K\"ahler cone on the two-dimensional 
plane spanned by the first Chern class and a certain null vector $\nv(S)$ (see \eqref{nullv}) 
that plays a crucial role in the construction.
Since for Hirzebruch surfaces $b_2(S)=2$, in this case the restriction is empty.

Thus, at least for Hirzebruch surfaces, which we denote by $\IF_m$, for arbitrary polarization $J$ we have explicit expressions 
for the generating functions of refined VW invariants and for their modular completions encoding the physical partition function,
as explained above. Both of them are expressed in terms of indefinite theta series with kernels constructed 
from the sign functions and, in the case of completions, the {\it generalized error functions} introduced in 
\cite{Alexandrov:2016enp,Nazaroglu:2016lmr} (see appendix \ref{ap-generr} for their definition). 
The same functions appear also in the original anomaly equations 
and thus are a crucial ingredient of the construction.

In this work, we elaborate on these results and analyze the following issues:

\begin{enumerate}
\item 
In \cite{Pioline:2025xgf} the authors made an interesting observation about the modular anomaly 
satisfied by the generating functions of refined D4-D2-D0 BPS indices in the special case of {\it collinear} D4-brane charges.
It is relevant to our story because 
on the non-compact CY ${\rm Tot}(K_S)$, used to obtain VW theory, 
there is a single compact divisor $[S]$ that can be wrapped by D4-brane.
It was shown that in such situation, at least for small D4-brane charges, the modular anomaly is significantly simplified 
being expressed in terms of {\it complementary generalized error functions}.
While the original construction represents this anomaly as a theta series with a kernel given by a sum over 
what are known as {\it Schr\"oder trees} with non-trivial weights assigned to their vertices, after the simplification
the kernel reduces to a single term, up to certain contributions localized on a small class of lattice vectors.
In fact, the new representation directly shows that the modular anomaly is given by an iterated integral 
of Eichler type as in \cite{Manschot:2017xcr}.
Then in \cite[Conjecture 6.1]{Alexandrov:2025sig}, on the basis of this observation,  
a formula for the simplified anomaly at any rank,
including the above mentioned special contributions, was conjectured. 
	
Here we find that the conjecture in \cite{Alexandrov:2025sig} was actually not precise
and prove instead its corrected version (see Theorem \ref{thm-simple-ref}).
As we will see, the main trouble is due to the special contributions which the original observation did not care about.
In fact, they are somewhat ambiguous because in \cite{Alexandrov:2016enp,Nazaroglu:2016lmr}
the complementary generalized error functions were defined only away from their loci of discontinuity, which 
are exactly the loci giving rise to these contributions.
We show that there are two natural ways to extend the definition of the error functions 
to these loci, which give two different expressions for the anomaly.
In particular, in the second expression the troubling contributions get canceled and do not appear at all,
so that the kernel of the theta series representing the anomaly does reduce to a single term 
(see Theorem \ref{thm-simple-ref2})!
	
\item 
As was mentioned above, VW invariants experience wall-crossing and therefore their generating functions
as well as their completions depend on polarization $J$, an element of the K\"ahler cone.
Whereas the wall-crossing behavior of the generating functions is governed by the wall-crossing formulae 
for the topological invariants \cite{ks,Joyce:2008pc}, it is interesting to understand the dependence on $J$ of their completions.
These functions are actually smooth and do not see any walls. Nevertheless, their dependence on $J$ is non-trivial
and we compute the difference of the completion $\whh_{N,\mu}$ evaluated at two different polarizations, $J$ and $J'$  
(see Theorem \ref{th-compl}).
The result is expressed in terms of the completions $\whh_{N_i,\mu_i}$ with $N_i<N$ evaluated at $J$ and is manifestly modular.
Furthermore, it takes a universal form independent of the surface data suggesting that it might hold also in a more general setting.

\item 
The wall-crossing formulae, at least in principle, allow us to compute 
(refined) VW invariants almost everywhere in the moduli space if they are known
at one point. The only exception is the walls themselves. We are not aware of any equations that give access to the invariants
evaluated at the walls.
This fact creates additional complications for the evaluation of VW invariants on different surfaces. 
Indeed, since any smooth rational surface can be obtained 
from $\IF_m$ by a finite sequence of blow-ups or blow-downs, 
VW invariants on such surfaces could be obtained by applying the blow-up formulae \cite{Yoshioka:1996,0961.14022,Li:1998nv}
to the generating functions on $\IF_m$. However, this requires evaluating them at the walls, which appears to be problematic.
In the usual approach (see, e.g., \cite{Manschot:2014cca}) 
this problem is avoided by passing through stack invariants \cite{joyce2008configurations}, 
which are polynomial combinations of rational VW invariants having simpler transformation properties under wall-crossing
(but more complicated integrality properties) so that they can be found directly at the walls.
Although conceptually straightforward, computationally this detour is certainly unwelcome.

Remarkably, in \cite{Alexandrov:2020dyy} it was argued that the expressions for $h_{N,\mu}$ found there 
in terms of indefinite theta series are well-defined everywhere including the walls!
In particular, they can be used in the blow-up relations correctly reproducing well-known results for $\IP^2$.
This suggests that the theta series correctly capture the VW invariants on the walls of marginal stability and 
raises questions regarding their properties. One can check that, in general, 
they are {\it not} integers.\footnote{The coefficients of the generating functions $h_{N,\mu}$
	that possess nice (mock) modular properties are in fact {\it rational} VW invariants. 
	But away from the walls they can always be converted to integer invariants by means of the relation \eqref{inv-ratref}.
	The integrality of \eqref{inv-ratref} represents a non-trivial consistency condition on the generating functions.} 
For $N=2$ they turn out to be simply the arithmetic average of the VW invariants on the two sides of the wall.
However, for $N>2$ this is not true anymore.
In the last part of the work, we evaluate the deviation of the generating functions of VW invariants at the walls
from the arithmetic average. The result is given by an iterative formula expressed through the generating functions of lower ranks 
(see Theorem \ref{th-wall})
and also hints that it might be a particular manifestation of a universal structure.
We hope that it will help better understand the meaning of these invariants.

\end{enumerate}

The organization of the paper is as follows.
In the next section we review the anomaly equation and the solution for the generating functions of refined VW invariants and their 
modular completions on Hirzebruch surfaces.
Then in \S\ref{sec-collinear} we reconsider the anomaly in the case of collinear D4-brane charges
and prove its simplified form.
In \S\ref{sec-varcompl} we compute the variation of the completion of the generating functions upon change of polarization.
Finally, in \S\ref{sec-VEwall} we compute the generating functions at the walls of marginal stability.
Several appendices contain useful facts about generalized error functions and details of some calculations.

\section{Modular anomaly and VW invariants on Hirzebruch surfaces}
\label{sec-anomaly}

\subsection{Anomaly equation}
\label{subsec-anomeq}

We start by recalling the anomaly equation satisfied by generating functions of refined D4-D2-D0 BPS indices.
These functions are defined by
\be 
\hr_{p,\mu}(\tau,z) = \sum_{\hq_0 \leq \hq_0^{\rm max}}
\frac{\bOm_{p,\mu}(\hq_0,y)}{y-y^{-1}}\,e^{-2\pi\I \hq_0 \tau },
\label{defhDTr}
\ee 
where $y=e^{2\pi\I z}$ is the refinement parameter and $\bOm_{p,\mu}(\hq_0,y)$ is a rational version (as \eqref{defcref} in the VW case)
of the Poincar\'e polynomial in $y$ of the moduli space of semi-stable coherent sheaves on a CY threefold $\CY$,
characterized by a charge vector $\gamma=(0,p^a,q_a,q_0)$, $a=1,\dots, b_2(\CY)$, with components corresponding 
to D6, D4, D2 and D0 charge, respectively. In particular, $p$ corresponds to a divisor $\cD_p=p^a\cD_a$ where $\cD_a$
is a basis in $H_4(\CY,\IZ)$ and therefore itself can be seen as an element of this homology group.
Similarly, $q\in H_2(\CY,\IZ)+\hf\, p$
where the lowering and raising indices is done by the metric $\kappa_{ab}=\kappa_{abc}p^c$
defined by the triple intersection numbers $\kappa_{abc}$ and the D4-brane charge.
This metric endows the lattice $\Lambda_p:=H_4(\CY,\IZ)$ with a quadratic from and ensures that $\Lambda_p^*=H_2(\CY,\IZ)$ is the dual lattice.
The embedding $\Lambda_p\hookrightarrow \Lambda_p^*$ provides the decomposition 
\be
q_a=\kappa_{ab}\eps^b +\mu_a+\hf\,\kappa_{ab}p^b,
\label{decomp-spfl}
\ee 
where $\eps\in\Lambda_p$, while $\mu\in \Lambda_p^*/\Lambda_p$ is the residue class 
taking $\det\kappa_{ab}$ independent values.
The BPS indices do not depend on the shift of D2-brane charge by an element of $\Lambda_p$ 
(the transformation known as spectral flow or tensoring with a line bundle). This is why they depend only on $p^a$, $\mu_a$
and a modification of the D0 charge remaining invariant under this transformation
\be
\hq_0=q_0-\hf\, \kappa^{ab}q_a q_b.
\label{def-hq0}
\ee 
The important fact which makes the function \eqref{defhDTr} well-defined is that $\hq_0$ is bounded from above.

The main result of \cite{Alexandrov:2019rth} is that the generating functions $\hr_{p,\mu}$ are 
(vector valued higher depth) mock Jacobi forms\footnote{The weight and index of the Jacobi forms are given by 
$$
w=-\hf\, \rank(\Lambda_p),
\qquad
m(p)=-\chi(\cO_{\cD_p})-\lambda_a p^a,
\label{index-Hr}
$$ 
where $\chi(\cO_{\cD_p})=\frac16\, \kappa_{abc}p^a p^b p^c+ \frac{1}{12}\, c_{2,a} p^a$ is the arithmetic genus of the divisor 
$\cD_p$ and $\lambda_a$ are integer parameters whose origin is not understood yet in general case.
We refer to \cite{Alexandrov:2025sig} for definitions and properties of mock modular and mock Jacobi forms.}, 
and their modular completion has the following form
\be
\whhr_{p,\mu}(\tau,\btau,z)=\hr_{p,\mu}(\tau,z)
+ \sum_{n=2}^{r} \frac{1}{2^{n-1}}
\sum_{\sum_{i=1}^n p_i=p}
\sum_{\bfmu}
\rmRirf{\bfp}_{\mu,\bfmu}(\tau, \btau,z)
\prod_{i=1}^n \hr_{p_i,\mu_i}(\tau,z),
\label{exp-whhr}
\ee
where the sum goes over ordered decompositions of the divisor $\cD_p$ into effective divisors, $r$ is the maximal number of components in 
these decompositions, and bold letters denote ordered sets of $n$ elements such as $\bfp=(p_1,\dots,p_n)$. The most important ingredient
is the coefficient $\rmRirf{\bfp}_{\mu,\bfmu}$ determining the non-holomorphic anomalous contribution:
\be
\rmRirf{\bfp}_{\mu,\bfmu}(\tau,\btau,z)
=
\sum_{q_i\in \Lambda_{p_i}+\mu_i+\hf  p_i \atop \sum_{i=1}^n q_i=\mu+\hf p}  
\Sym \Bigl\{ (-y)^{\sum_{i<j} \gamma_{ij}}\, \scRrf_n(\bfhgam;\tau_2,\beta)
\Bigr\} \,e^{\pi\I \tau Q_n(\bfhgam)},
\label{Rirf-to-rmRrf}
\ee
where $\Sym$ denotes symmetrization\footnote{The symmetrization could be omitted since it is automatically implemented by the sum over 
decompositions of $p$ in \eqref{exp-whhr}, but for our purposes it will be convenient to keep it.} 
(with weight $1/n!$) with respect to permutation of charges $\hgam_i=(p_i^a,q_{i,a})$,
$\gamma_{ij}$ is the Dirac pairing
\be 
\gamma_{ij}=q_{i,a}p_j^a-q_{j,a}p_i^a,
\label{Dirac}
\ee 
and we represented the refinement parameter as $z=\alpha-\tau\beta$ with $\alpha,\beta\in \IR$.
The coefficient \eqref{Rirf-to-rmRrf} has the form a theta series over the lattice 
\be 
\bfLam^{(\bfp)} =
\(\oplus_{i=1}^n \Lambda_{p_i}\)/\Lambda_p,
\label{latp}
\ee
endowed with the quadratic form\footnote{$\bfq=\bigl(\kappa_1^{ab} q_{1,b} ,\dots ,\kappa_n^{ab}q_{n,b}\bigr)$ 
	can be seen as an element of the dual lattice $(\bfLam^{(\bfp)})^*$ 
	up to addition of the unit vector $\boldsymbol{1}$. 
	This is why we allow the sum of its components to be non-zero.
	See also appendix \ref{ap-lat}.} 
\be
Q_n(\bfhgam)= \kappa^{ab}q_a q_b-\sum_{i=1}^n\kappa_i^{ab}q_{i,a} q_{i,b} \, ,
\qquad
q_a=\sum_{i=1}^n q_{i,a},
\label{defQlr}
\ee
where $\kappa_{i,ab}=\kappa_{abc}p_i^c$ and $\kappa_i^{ab}$ is its inverse, with the kernel given by $\scRrf_n$.
Note that the quadratic form \eqref{defQlr} has $n-1$ negative directions and therefore 
the kernel must be non-trivial for the theta series to be convergent.

To construct the kernel, we need to introduce several objects.
Most of them will also appear in the next subsection in the expression for the generating functions of VW invariants.
First, we define a set of vectors $\bfv_{ij}\in\bfLam^{(\bfp)}$ by specifying their components
\be 
(\bfv_{ij})_k^a=\delta_{ki} p_j^a-\delta_{kj} p_i^a,
\label{def-vij}
\ee
as well as two other sets
\be
\bfv_\ell= \sum_{i=1}^\ell\sum_{j=\ell+1}^n\bfv_{ij},
\qquad
\bftet = \sum_{i<j} \bfv_{ij}.
\label{def-bfvk}
\ee
The meaning of these vectors is very simple: they allow us to get the Dirac pairing \eqref{Dirac} upon contraction with 
the vector of electric charges 
$\bfq=\bigl(\kappa_1^{ab} q_{1,b} ,\dots ,\kappa_n^{ab}q_{n,b}\bigr)$ 
by means of the bilinear form
\be
\bfx\ast\bfy=\sum_{i=1}^n \kappa_{i,ab}x_i^a y_i^b.
\label{biform}
\ee
Indeed, it is easy to check that 
\be
\bfv_{ij}\ast\bfq=\gamma_{ij},
\qquad 
\Gamma_\ell:=\bfv_\ell \ast\bfq=\sum_{i=1}^\ell\sum_{j=\ell+1}^n\gamma_{ij},
\qquad
\bftet\ast\bfq=\sum_{i<j}\gamma_{ij}.
\label{defGammae}
\ee

These vectors allow us to introduce two sets of functions
\bea
\Er_n(\bfhgam;\tau_2,\beta)&=& \Phi^E_{n-1}\(\{ \bfv_{\ell}\};\sqrt{2\tau_2}\,(\bfq+\beta\bftet )\),
\label{Erefsim}
\\
\Efrf_n(\bfhgam)&=& S(\Gamma_{\Zv_{n-1}}):=
\sum_{\cI\subseteq \Zv_{n-1}} e_{|\cI|}\,\delta(\Gamma_\cI)\, 
\sgn(\Gamma_{\Zv_{n-1}\setminus \cI}),
\label{defSlin}
\eea
where 
$\Phi^E_n$ are the generalized error functions defined in \eqref{generrPhiME},
$\Zv_{n}=\{1,\dots,n\}$, $\Gamma_\cI=\{\Gamma_\ell\}_{\ell\in\cI}$, and\footnote{$e_n$ 
	are the Taylor coefficients of $\mbox{Arctanh}(x)/x$. More on their properties, see appendix \ref{ap-coef}.}
\be
e_n=\left\{\begin{array}{ll}
	0 & \mbox{\rm if $n$ is odd},
	\\
	\frac{1}{n+1}\ & \mbox{\rm if $n$ is even},
\end{array}\right.
\label{def-genthm}
\ee  
\be 
\delta(x_\cI)=\prod_{\ell\in\cI}\delta_{x_\ell},
\qquad
\sgn(x_\cI)=\prod_{\ell\in \cI}\sgn(x_\ell).
\label{def-delta}
\ee
The functions $\Efrf_n$ are related to the large $\tau_2$ limit of $\Er_n$ evaluated at $\beta=0$.
However, they are not equal to the limit directly, but require an additional symmetrization \cite{Alexandrov:2024jnu}\footnote{In fact, 
	there is a stronger relation \cite{Alexandrov:2025sig}
	$$ 
	\Sym\Bigl\{y^{\sum_{i<j} \gamma_{ij}}\lim_{\tau_2\to\infty}\Er_n(\bfhgam;\tau_2,0)\Bigr\}
	=\Sym\Bigl\{y^{\sum_{i<j} \gamma_{ij}}\Efrf_n(\bfhgam) \Bigr\}.
	$$
	\label{foot-limy}}
\be 
\Sym\lim_{\tau_2\to\infty}\Er_n(\bfhgam;\tau_2,0)=\Sym\Efrf_n(\bfhgam) .
\label{Efref-lim}
\ee
Then the difference $\Eprf_n:=\Er_n-\Efrf_n$ represents (at $\beta=0$ and after symmetrization) 
the exponentially suppressed part of $\Er_n$.

\begin{figure}[t]
	\begin{center}
		\includegraphics[width=5.7cm]{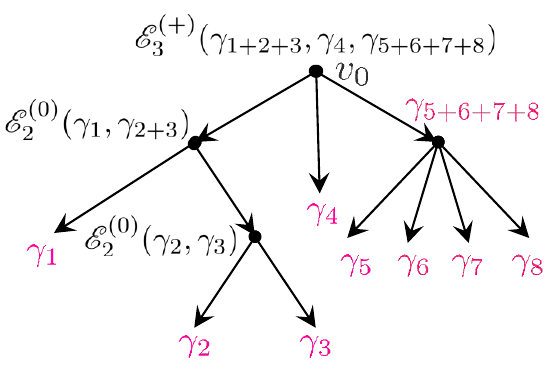}
	\end{center}
	\vspace{-0.9cm}
	\caption{An example of Schr\"oder tree contributing to $\scRrf_8$ and 
	the weights assigned to its vertices by \eqref{refsolRn}. The charges are shown
	using the shorthand notation $\gamma_{i+j}=\gamma_i+\gamma_j$.
	\label{fig-WRtree}}
\end{figure}

Finally, let $\IT_n^{\rm S}$ be the set of the Schr\"oder trees with $n$ leaves, i.e. rooted planar trees such that 
all vertices $\vert\in V_T$ (the set of vertices of $T$ excluding the leaves) have $n_\vert\geq 2$ children.
Their $n$ leaves carry charges $\hgam_i$, whereas the charges assigned to other vertices
are given recursively by
the sum of charges of their children $\Ch(\vert)$, i.e. $\hgam_\vert=\sum_{\vert'\in\Ch(\vert)}\hgam_{v'}$.
Besides, we take $n_T$ to be the number of elements in $V_T$ and $\vert_0$ to denote the root vertex.
Then, given a Schr\"oder tree $T$,
we set $\Ev_{\vert}\equiv \Ev_{n_\vert}(\{\hgam_{\vert'}\})$ (and similarly for $\Ef_{\vert}, \Ep_{\vert}$
and many similar functions that will appear below)
where $\vert'\in \Ch(\vert)$ runs over the $n_\vert$ children of the vertex $\vert$ (see Fig. \ref{fig-WRtree}). 
Using these notations, we can finally write a formula for the kernel $\scRrf_n$:
\be
\scRrf_n\(\bfhgam;\tau_2,\beta\) = \sum_{T\in\IT_n^{\rm S}}(-1)^{n_T-1} 
\Eprf_{v_0} \prod_{v\in V_T \backslash \{v_0\}}\Efrf_v.
\label{refsolRn}
\ee

A few comments are in order:
\begin{itemize} 
\item 
The mock modularity of the generating functions $\hr_{p,\mu}$ and the equation \eqref{exp-whhr} for their completion 
do not hold everywhere in the moduli space, but only in the large volume attractor chamber \cite{deBoer:2008fk}.

\item 
Note the presence of the Kronecker symbols $\delta(\Gamma_\cI)$ in \eqref{defSlin}. 
Their effect is such that, if $m$ of $\Gamma_\ell$'s vanish, 
$\Efrf_n$ is given by $e_m$ multiplied by the product of $n-m-1$ sign functions of non-vanishing $\Gamma_\ell$'s. 
In general, the Kronecker symbols do contribute and
generate those special contributions that were mentioned in the introduction.

\item 
In this work we will not need the above construction in full generality.
The case we are really interested in is when all D4-brane charges are collinear, i.e. $p^a=r p_0^a$ with fixed 
irreducible $p_0^a$.
Then the sum over Schr\"oder trees in \eqref{refsolRn} can be explicitly evaluated
and  the kernel crucially simplifies, as will be shown in \S\ref{sec-collinear}. 
\end{itemize}

\subsection{Generating functions of refined VW invariants}
\label{subsec-genfunVW}

Refined VW invariants of surface $S$ are defined by Poincar\'e polynomials of the moduli spaces $\cM_{\gamma,J}$
of semi-stable coherent sheaves on $S$, namely,
\be
\Omega_J(\gamma,y)= \sum_{p=0}^{2d_{\IC}(\cM_{\gamma,J})} y^{p-d_{\IC}(\cM_{\gamma,J})}\, b_p(\cM_{\gamma,J}),
\label{intOm}
\ee
where $d_{\IC}(\cM)$ is the complex dimension of $\cM$ and $b_p(\cM)$ is its Betti number.
Both the moduli space and the invariants are parametrized by charge $\gamma=(N,\mu,\hf\mu^2-n)$
and polarization $J$ determining the stability condition. (See, e.g., \cite{Beaujard:2020sgs}
	and references therein for a more rigorous and detailed exposition.)
From physics point of view, $\cM_{\gamma,J}$ parametrizes solutions of hermitian Yang-Mills equations with gauge group $SU(N)$, 
while $J$ is the K\"ahler form appearing in the self-duality condition on the field strength $F$.
The charge $\gamma$ can be seen as the Chern character of the sheaf, in particular, 
$\mu=-c_1(F)\in \Lambda_S:=H^2(S,\IZ)$ and $n=\int_S c_2(F)\in \IZ$ are Chern classes of the gauge bundle.
In fact, the moduli space does not change upon tensoring $F$ with a line bundle $\cL$, 
which leads to the shift $\mu\to \mu-Nc_1(\cL)$, but leaves invariant the rank $N$ and
the Bogomolov discriminant
\be
\Delta(F):= \frac{1}{N} \left( n - \frac{N-1}{2N} \mu^2 \right),
\label{Bogom}
\ee
where $\mu^2= \int_S\mu^2$.
Due to this, the parameter $\mu$ can be restricted to $\Lambda_S/N\Lambda_S$.

At this point, the parallel to the BPS indices discussed in the previous subsection should become clear:
$N$ corresponds to D4-brane charge $p^a$, while $N\Delta(F)$ is analogous to $-\hq_0$ up to a constant shift.
The precise counterpart of the generating functions \eqref{defhDTr} is\footnote{To avoid cluttering, we will not use superscript ``ref"
	for the generating functions of refined VW invariants.} 
\be
\label{defhVWref}
h_{N,\mu,J}(\tau,z) =
\sum_{n\geq 0}
\frac{\bOm_J(\gamma,y)}{y-y^{-1}}\,
\q^{N\(\Delta(F) - \tfrac{\chi(S)}{24}\)},
\ee
where we introduced the standard notation $\q=e^{2\pi\I\tau}$, $\chi(S)$ is the Euler number of the surface, 
and we defined the rational refined invariants
\cite{joyce2008configurations,Manschot:2010qz}
\be
\bOm_J(\gamma,y) =  \sum_{d|\gamma} \frac{(-1)^{d-1}(y-y^{-1})}{d(y^d-y^{-d})}\, \Omega_J(\gamma/d, - (-y)^{d}).
\label{defcref}
\ee
Only the generating functions of the rational invariants, rather than of the integer ones,
manifest nice modular properties \cite{Manschot:2010xp}.
Note that the inverse relation to \eqref{defcref} is given by (see, e.g., \cite{Pioline:2025xgf})
\be
\Omega_J(\gamma,y) =  \sum_{d|\gamma} \frac{\upmu(d)}{d}\,\frac{(-1)^{d-1}(y-y^{-1})}{y^d-y^{-d}}\,  \bOm_J(\gamma/d, - (-y)^{d}),
\label{inv-ratref}
\ee
where $\upmu(d)$ is the multiplicative M\"obius function
\be 
\upmu(d)=\left\{\begin{array}{ll}
	1, & d=1,
	\\
	(-1)^k, \quad & d\mbox{ is a product of }k\mbox{ distinct primes},
	\\
	0, & \mbox{otherwise}.
\end{array}\right.
\ee 
Therefore, the generating function of integer refined VW invariants can be obtained as 
\be 
\hint_{N,\mu,J}(\tau,z)=\sum_{d|(N,\mu)} \frac{\upmu(d)}{d}\, h_{\frac{N}{d},\frac{\mu}{d},J}(d\tau,z_d),
\qquad 
z_d:=d\(z+\hf\)-\hf\, .
\label{def-hint}
\ee

As shown in \cite{Alexandrov:2019rth}, provided $b_2^+(S)=1$ and $b_1(S)=0$, the generating functions of VW invariants 
$h_{N,\mu,J}$ can be identified with $\hr_{p,\mu}$
from the previous subsection, and hence satisfy the same modular anomaly equation, 
upon choosing $\CY$ to be the total space of the canonical bundle over 
the surface, $\CY={\rm Tot}(K_S)$, which is assumed to have a single section. 
In this case:
\begin{itemize}
\item 
$b_2(\CY)=b_2(S)+1$;

\item 
the only admissible D4-brane charges are $p^a=Np_0^a$, where $p_0^a$ is the charge corresponding to
divisor $[S]$, which in a suitable basis has components $p_0^a=(1,-c_1^\alpha(S))$, $\alpha=1,\dots,b_2(S)$;
	
\item 
in this basis, D2-brane charges are restricted to have vanishing first component, i.e. $q_a=(0,q_\alpha)$;

\item 
the quadratic form $\kappa_{ab}$ is given by
\be 
\kappa_{ab}=N\(\begin{array}{cc}
	0 & 0 \\ 0 & C_{\alpha\beta}
	\end{array}\), 
\ee 
where $C_{\alpha\beta}$ is the intersection matrix on $S$;

\item 
the large volume attractor chamber in the K\"ahler moduli space of $\CY$ corresponds to the canonical chamber 
containing the point $J=c_1(S)$ in the space of polarizations.
\end{itemize}  
These data allow us immediately to rewrite the anomaly equation \eqref{exp-whhr} and all its ingredients 
in terms of the surface data.
For Hirzebruch and del Pezzo surfaces, the resulting equation has been solved for $h_{N,\mu,c_1}$ 
and their completions $\whh_{N,\mu,c_1}$ in \cite{Alexandrov:2020bwg}.
The solution has been generalized to arbitrary $J$ belonging to a two dimensional plane in \cite{Alexandrov:2020dyy}.
In the rest of this subsection, we present this solution restricting to Hirzebruch surfaces $\IF_m$, $m=0,1,2$.

To this end, let us recall that $\IF_m$ is a projectivization of the $\cO(m)\oplus \cO(0)$ bundle over $\IP^1$.
It has $b_2(\IF_m)=2$ and in the basis given by the curves corresponding to the fiber $[f]$ and the section $[s]$ of the bundle,
the intersection matrix and the first Chern class are the following
\be
C_{\alpha\beta} = \begin{pmatrix} 0 & 1 \\ 1 & -m \end{pmatrix},
\qquad
c_1(\IF_m) = (m+2)[f]+2[s].
\label{dataFk}
\ee
The intersection matrix is unimodular, has signature (1,1), and induces a bilinear form on the lattice $\Lambda_S=H^2(S,\IZ)$,
which will be denoted as $x\cdot y$.
Similarly to \eqref{latp}, we consider the lattice 
\be 
\bfLam^{(\bfN)}_S =
\(\oplus_{i=1}^n N_i\Lambda_S\)/(N\Lambda_S),
\label{latVW}
\ee
where $\bfN=(N_1,\dots,N_n)$ is a $n$-tuple of positive integers and $N=\sum_{i=1}^n N_i$.
With this lattice, we associate a theta series (cf. \eqref{Rirf-to-rmRrf})\footnote{The physical charges differ from 
	$q_i$ by the presence of the additional term $-\hf N_i c_1(S)$ coming from the last term in \eqref{decomp-spfl}.
	However, $q_i$ enter the construction only through the Dirac products \eqref{gam12} where this term cancels.
	Due to this reason, we omit it from the very beginning.}
\be 
\vth^{(\bfN)}_{\mu,\bfmu}(\tau,z;\cF_n)=
\sum_{q_i\in N_i\Lambda_S+\mu_i\atop \sum_{i=1}^n q_i=\mu}
\cF_n(\bfhgam)\,\q^{\hf\, Q_n(\bfhgam)}\, y^{\sum_{i<j}\gamma_{ij}(c_1(S))},
\label{Theta-bfN}
\ee 
where $\mu_i\in \Lambda_S/N_i\Lambda_S$ are residue classes, 
$\hgam_i=(N_i,q_{i,\alpha})$ are relevant charges, 
$Q_n(\bfhgam)$ is the same quadratic form as \eqref{defQlr} which now takes the form
\be
Q_n(\bfhgam)=
\frac{1}{N}\, q^2-\sum_{i=1}^n \frac{1}{N_i} q_{i}^2
=-\sum_{i<j}\frac{(N_i q_j - N_j q_i)^2}{NN_i N_j}
\label{defQlr-VW}
\ee
with $q^2=C^{\alpha\beta}q_\alpha q_\beta$ and $C^{\alpha\beta}$ the inverse of $C_{\alpha\beta}$,
$\gamma_{ij}(J)$ denotes (cf. \eqref{Dirac})
\be
\label{gam12}
\gamma_{ij}(J)=
J^{\alpha} (N_i q_{j,\alpha} -N_j q_{i,\alpha}),
\ee
and $\cF_n(\bfhgam)$ is a kernel to be specified, which itself can depend (even non-holomorphically) on $\tau$ and $z$.
Note that the signature of $Q_n$ is $(n-1,n-1)$ and thus $\vth^{(\bfN)}_{\mu,\bfmu}$ is an indefinite theta series.

Another important ingredient is given by the generating functions $H^S_{N,\mu}$ of stack invariants briefly mentioned in 
the introduction. We are interested in $H^S_{N,\mu}$ evaluated at $J=\nv\in\Lambda_S$, 
a specific null vector in the sense that $\nv^2=0$, which for $S=\IF_m$ is given by 
\be
\nv(\IF_m)=[f].
\label{nullv}
\ee
At this point in the moduli space, these functions read \cite{Manschot:2011ym,Mozgovoy:2013zqx}
\be
H^{\IF_m}_{N,\mu}=\delta^{(N)}_{\nv\cdot\mu} \, H_N ,
\qquad
H_N=\frac{\I (-1)^{N-1} \eta(\tau)^{2N-3}}
{\theta_1(\tau,2Nz)\, \prod_{m=1}^{N-1} \theta_1(\tau,2mz)^2}\, ,
\label{gfHN}
\ee
where $\eta(\tau)$ is the Dedekind eta function, $\theta_1(\tau,z)$ is the Jacobi theta function, and we used the convenient
notation 
\be
\delta^{(n)}_x=\left\{\begin{array}{ll}
1\ \ & \mbox{if }x=0\!\! \mod n,
\\
0 & \mbox{otherwise.}
\end{array}\right. 
\label{deltan}
\ee 

Given these definitions, the generating functions of refined VW invariants $h_{N,\mu,J}$ and their completions $\whh_{N,\mu,J}$
can be presented in a unified way as 
\bea
h_{N,\mu,J}(\tau,z)&=&\Theta_{N,\mu}\Bigl(\tau,z;\{\Fv_n(J)\},\{H^S_{N',\mu}\}\Bigr),
\label{genJhN}
\\
\whh_{N,\mu,J}(\tau,z)&=&\Theta_{N,\mu}\Bigl(\tau,z;\{\whFv_n(J)\},\{H^S_{N',\mu}\}\Bigr),
\label{complFBJ}
\eea
where we introduced a contraction of the theta series \eqref{Theta-bfN} with a set of vector valued functions
\be 
\Theta_{N,\mu}\Bigl(\tau,z;\{\cF_n\},\{H_{N',\mu}\}\Bigr)=
\sum_{n=1}^N \frac{1}{2^{n-1}}\sum_{\sum_{i=1}^n N_i=N}\sum_{\bfmu}
\vth^{(\bfN)}_{\mu,\bfmu}(\tau,z;\cF_n)
\prod_{i=1}^n H_{N_i,\mu_i}(\tau,z).
\label{def-Theta}
\ee 
Such contraction represents a common structure for most relevant quantities and will often appear below. 
In particular, the r.h.s. of the anomaly equation \eqref{exp-whhr} also has the same form upon identifying D4-brane charges $p_i^a$ 
with ranks $N_i$. Furthermore, $\Theta_{N,\mu}$ satisfies the important composition property 
\be \\
\Theta_{N,\mu}\Bigl(\{\cF_n\},\{\Theta_{N',\mu}(\{\cG_n\},\{H_{N'',\mu}\})\}\Bigr)
=\Theta_{N,\mu}\Bigl(\{\cF^{\rm tot}_n\},\{H_{N',\mu}\}\Bigr),
\label{compTheta}
\ee  
where
\be 
\cF^{\rm tot}_n(\bfhgam)=\sum_{m=1}^n \sum_{\sum_{s=1}^{m} n_s=n}\cF_m(\{\hbbgam_s\})\prod_{s=1}^m\cG_{n_s}(\bfhgam^{(s)}),
\label{compos-ker}
\ee 
and we introduced notations for the subsets of charges induced by the decomposition of $n$:
$\bfhgam^{(s)}=\{\hgam_i\}$ denotes the $s$-th subset of size $n_s$ 
and $\hbbgam_s=\sum_i \hgam^{(s)}_i$ is the total charge of this subset.\footnote{Of course, both $\bfhgam^{(s)}$ and $\hbbgam_s$
	depend on the ordered partition $\{n_s\}$. We do not indicate this dependence explicitly to avoid cluttering.}

To complete the presentation, it remains to specify the two kernels, 
$\Fv_n(J)$ and $\whFv_n(J)$. For $n=1$ they are set to 1, while for $n>1$
both of them involve the following combination of charges determined by the null vector $\nv(S)$ \eqref{nullv}
\be 
\Bv_\ell=\gamma_{\ell,\ell+1}(\nv(S))+\beta N_\ell N_{\ell+1}(N_\ell+N_{\ell+1})\, \nv(S)\cdot c_1(S),
\ee 
where, as in the previous subsection, $\beta$ controls the imaginary part of the refinement parameter as $z=\alpha-\tau\beta$.
Then the kernel defining the generating functions is given by 
\be
\begin{split}
\Fv_n(\bfhgam;J)
=&\,
\sum_{\cI\subseteq \Zv_{n-1}} e_{|\cI|}\,\delta(\Gamma_\cI(J))
\prod_{\ell\in \Zv_{n-1}\setminus \cI}\Bigl(\sgn(\Gamma_\ell(J))-\sgn(\Bv_\ell)\Bigr)
\\
=&\,
\sum_{\cI\subseteq \Zv_{n-1}}S(\Gamma_\cI(J))\,\sgn\bigl(-\Bv_{\Zv_{n-1}\setminus \cI}\bigr),
\end{split} 
\label{kerg}
\ee
where we used the notations from \eqref{defSlin}-\eqref{def-delta} and 
\be
\Gamma_\ell(J):=\sum_{i=1}^\ell\sum_{j=\ell+1}^n\gamma_{ij}(J).
\label{defGammaVW}
\ee
The kernel defining the modular completion of the generating functions can be represented\footnote{In \cite{Alexandrov:2020bwg,Alexandrov:2020dyy}
	another representation was used for this kernel which is based on a sum over subsets as in \eqref{kerg}.
	However, it is easy to see that the two representations are equivalent (see the next subsection for more on this equivalence)
	and the one in \eqref{kerhg-tree} is more convenient for our purposes.}
	 as a sum over
planar rooted trees with $n$ leaves all having depth 2 (the property known as uniform leaf depth). 
The set of these trees will be denoted by $\IT^{\rm r,2}_n$. Then we have
\be
\whFv_n(\bfhgam;J)=\sum_{T\in \IT^{\rm r,2}_n} 
\Ev_{v_0}(J)\prod_{v\in V_T\setminus\{v_0\}}\sgn\bigl(-\Bv_{\Ch'(v)}\bigr).
\label{kerhg-tree}
\ee
Here the charges are assigned to the vertices and leaves of the tree following the same rules as for the Schr\"oder trees 
and we used the same notations as in \eqref{def-delta} and \eqref{refsolRn}, except that
$\Ch'(v)$ denotes the set of children of vertex $v$ with the rightmost child omitted (see Fig. \ref{fig-trees-r2}). 
The only new thing is that the vectors $\bfv_{ij}$, $\bfv_\ell$, etc., appearing in the generalized error function \eqref{Erefsim}
are defined by D4-brane charges relevant for VW theory and thus take the form (cf. \eqref{def-vij}, \eqref{def-bfvk})
\be 
(\bfv_{ij})_k^\alpha=J^\alpha(\bffv_{ij})_k,
\qquad 
(\bfv_\ell)_k^\alpha=J^\alpha(\bffv_\ell)_k,
\label{bfvVW}
\ee
where
\be
(\bffv_{ij})_k=\delta_{jk} N_i-\delta_{ik} N_j,
\qquad 
\bffv_\ell= \sum_{i=1}^\ell\sum_{j=\ell+1}^n\bffv_{ij},
\label{def-bffvk}
\ee
while the bilinear form \eqref{biform} is replaced by
\be
\bfx\ast\bfy=\sum_{i=1}^n N_i x_i\cdot y_i.
\label{biformVW}
\ee
 
\begin{figure}[t]
	\begin{center}
		\includegraphics[width=18.5cm]{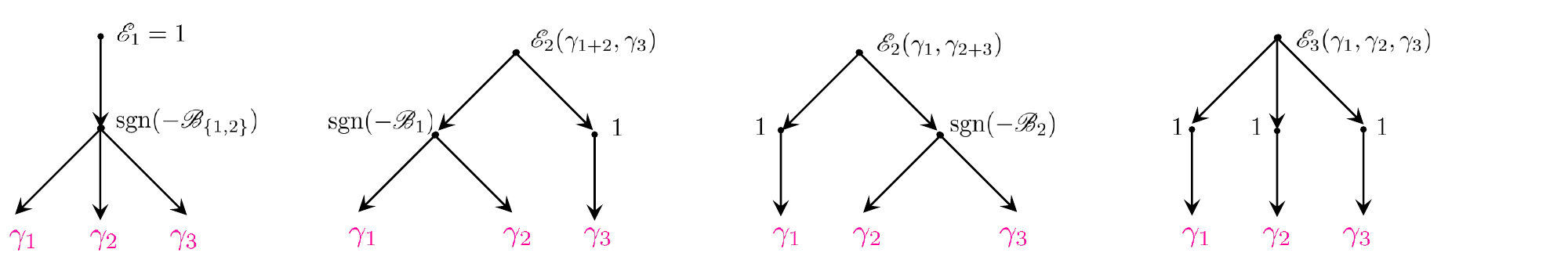}
	\end{center}
	\vspace{-1cm}
	\caption{Trees from $\IT^{\rm r,2}_3$ and the weights assigned to their vertices by \eqref{kerhg-tree}.
		\label{fig-trees-r2}}
\end{figure}

\subsection{Collinear charges and trees}
\label{subsec-useful}

Before we proceed with the facts we want to prove, it is useful to make a few technical observations on typical structures appearing 
in the above equations and the ones to follow.

As explained above, in the VW case D4-brane charges are restricted to be collinear $p^a=r p_0^a$.
Such restriction is responsible for many simplifications. For example, in \cite{Alexandrov:2019rth} it was proven 
that in this case the holomorphic anomaly of the completion, $\p_{\btau}\whhr_{p,\mu}$, gets only contributions 
proportional to $\whhr_{p_1,\mu_1}\whhr_{p_2,\mu_2}$ with $p_1+p_2=p$, whereas in general case all decompositions of $p$, 
with the number of constituents up to $r$, contribute.

Here we show another important simplification. It concerns the vectors $\bfv_\ell$ \eqref{def-bfvk} and relies on the following
\begin{lemma}
	\label{lemma1}
	Let $\cI\subset \Zv_{n-1}$ and $i\notin \cI$. We denote $i_-=\max\{j\in \cI\cup\{0\}:\ j<i\}$
	and $i_+=\min\{j\in \cI\cup\{n\}:\ j>i\}$. Besides, $\bfv_{i\perp \cI}$ denotes the projection of $\bfv_i$ 
	orthogonal to the subspace spanned by $\{\bfv_j\}_{j\in\cI}$.	
	Then in the collinear case where $p_i=r_ip_0$, one has 
	\be
	\bfv_{\ell\perp \cI}=\lambda \sum_{i=i_-+1}^\ell\sum_{j=\ell+1}^{i_+}\bfv_{ij},
	\qquad 
	\lambda=\frac{r}{\sum_{l=i_-+1}^{i_+}r_l}\, .
	\ee
\end{lemma}
\begin{proof}
	It is straightforward to calculate that 
	\be
	\begin{split}
		\bfv_{\ell\perp k} =&\, \frac{r}{\sum_{l=1}^k r_l}\sum_{i=1}^\ell \sum_{j=\ell+1}^k \bfv_{ij},
		\qquad\quad \
		k>\ell,
		\\  
		\bfv_{\ell\perp k} =&\, \frac{r}{\sum_{l=k+1}^n r_l}\sum_{i=k+1}^\ell \sum_{j=\ell+1}^n \bfv_{ij},
		\qquad 
		k<\ell. 
	\end{split}
	\label{iperpj}
	\ee 
	Since $\bfv_\ell\ast\bfv_{ij}=0$ for $\ell\ne i,j$, 
	this result ensures that $\bfv_{\ell\perp i}\ast \bfv_j=0$ both for $\ell<i<j$ and for $j<i<\ell$, and thus already proves the claim
	in the cases of $i_-=0$ and $i_+=n$. In other cases, we compute 
	\be
		\bfv_{\ell\perp \{i_-,i_+\}} = \frac{r}{\sum_{l=i_-+1}^{i_+}r_l}\sum_{i=i_-+1}^\ell\sum_{j=\ell+1}^{i_+}\bfv_{ij}.
	\ee 
	Since this result satisfies $\bfv_{\ell\perp \{i_-,i_+\}}\ast \bfv_j=0$ for $j<i_-$ and $j>i_+$, we conclude that
	$\bfv_{\ell\perp \cI}=\bfv_{\ell\perp \{i_-,i_+\}}$, which completes the proof.
\end{proof}

Let us explain the meaning of this lemma. The vectors $\bfv_\ell$ are constructed from $n$ charges $p_i$. 
A subset $\cI=\{j_1,\dots j_{\mI-1}\}\subset\Zv_{n-1}$ induces two related partitions  
into $\mI=|\cI|+1$ disjoint (possibly empty in the second case) subsets
\be  
\begin{array}{rlrl}
\Zv_n=& \cup_{s=1}^{\mI} \cI_s, 
\qquad&	 
\cI_s=&\{j_{s-1}+1,\dots,j_s\}, 
\\
\Zv_{n-1}=&\(\cup_{s=1}^{\mI} \cI'_s\)\cup\cI,
\qquad&
\cI'_s=&\{j_{s-1}+1,\dots,j_s-1\},
\end{array} 
\label{subsets}
\ee 
where we set $j_0=0$ and $j_{\mI}=n$. 
Then Lemma \ref{lemma1} tells us that, upon orthogonalization with respect to the vectors from the subset $\cI$,
$\bfv_\ell$ for $\ell\notin\cI$
becomes proportional to a similar vector constructed from the subset of charges $\{p_i\}_{i\in\cI_s}$ 
specified by $s$ such that $\ell\in\cI'_s$:
\be 
\bfv_{\ell\perp \cI}\sim \bfv^{(s)}_\ell:=\sum_{i=j_{s-1}+1}^\ell\sum_{j=\ell+1}^{j_s}\bfv_{ij}.
\label{rel-bfvperp}
\ee  
On the other hand, for $\ell\in\cI$ it is easy to see that $\Gamma_\ell$ depends only on the total charges
of the subsets, $\hbbgam_s=\sum_{i\in\cI_s}\hgam_i$.
Namely, introducing $\bbv_{st}$ defined by $\bbp_s$ similarly to \eqref{def-vij}, one has
\be 
\bfv_{j_s}=\bbv_s:=\sum_{t=1}^s\sum_{t'=s+1}^{\mI} \bbv_{tt'}.
\label{def-tbfvs}
\ee

These two facts are particularly useful because they allow us to rewrite sums over subsets as sums over certain types of trees.
As a simple example, let us consider $\Efrf_n(\bfhgam)$ defined in \eqref{defSlin}.
Since $\Gamma_\ell=\bfv_\ell\ast\bfq$, provided $\Gamma_i=0$, one has $\Gamma_\ell=\Gamma_{\ell\perp i}:=\bfv_{\ell\perp i}\ast \bfq$.
An obvious generalization of this fact is the relation
\be 
\delta(\Gamma_\cI)\Gamma_\ell=\delta(\Gamma_\cI)\Gamma_{\ell\perp\cI},
\label{rel-delperp}
\ee 
where $\Gamma_{\ell\perp\cI}=\bfv_{\ell\perp \cI}\ast \bfq$.
Then using \eqref{rel-bfvperp} and the fact that the proportionality coefficient in this relation is positive, one obtains 
\be 
\Efrf_n(\bfhgam)=
\sum_{\cI\subseteq \Zv_{n-1}} e_{|\cI|}\,\delta(\Gamma_\cI)\, 
\prod_{s=1}^{\mI} \sgn\bigl(\Gamma^{(s)}_{\cI'_s}\bigr),
\label{Efrf-new}
\ee
where $\Gamma^{(s)}_\ell=\bfv^{(s)}_\ell\ast \bfq$.
The fact that each factor in the product depends only on a subset of charges, 
while the first factor depends only on the total charges of the subsets,
allows us to rewrite \eqref{Efrf-new} as 
\be 
\Efrf_n(\bfhgam)=\sum_{T\in \IT^{\rm r,2}_n} 
e_{n_{v_0}-1}\,\delta_{v_0}\prod_{v\in V_T\setminus\{v_0\}}S^{(0)}_v,
\label{Efrf-tree1}
\ee
where $\IT^{\rm r,2}_n$ is the same set of planar rooted trees of uniform leaf depth 2 as in \eqref{kerhg-tree} 
and we denoted
\be 
\delta_v=\prod_{\ell=1}^{n_v-1}\delta_{\Gamma_\ell},
\qquad
S^{(0)}_v=\prod_{\ell=1}^{n_v-1}\sgn(\Gamma_\ell),
\ee 
meaning, as usual, that all $\Gamma_\ell$ here are constructed from the charges of children of vertex $v$.

The representation \eqref{Efrf-tree1} is however somewhat redundant because, 
as soon as a vertex has one child, its weight is equal to 1.
See, e.g., the situation in Fig. \ref{fig-trees-r2}. 
Such vertices can be removed by imposing the condition that the trees should be of Schr\"oder type.
As a result, we can equally write
\be 
\Efrf_n(\bfhgam)=\sgn(\Gamma_{\Zv_{n-1}}) 
+\sum_{T\in \IT^{\rm S,\leq 2}_n} 
e_{n_{v_0}-1}\,\delta_{v_0}\prod_{v\in V_T\setminus\{v_0\}}S^{(0)}_v,
\label{Efrf-tree2}
\ee
where $\IT^{\rm S,\leq 2}_n$ is the set of Schr\"oder trees of depth at most 2 and
the first term arises from the tree in \eqref{Efrf-tree1} with $n_{v_0}=1$.

Finally, it is useful to note that, from the discussion around \eqref{subsets}, it should be clear that
the sum over subsets $\cI\subseteq\Zv_{n-1}$ is also equivalent to the sum over ordered partitions of $n$,
namely,
\be 
\sum_{\cI\subseteq\Zv_{n-1}} =\sum_{m=1}^n\sum_{\sum_{s=1}^{m} n_s=n}
\label{ident-sumset}
\ee 
with the obvious identification $n_s=|\cI_s|$.
Thus, in total, we have four different ways to present quantities like $\Efrf_n$, 
and below we will switch from one to another depending on which one is more convenient for our purposes.

\section{Modular anomaly for collinear charges}
\label{sec-collinear}

In this section we prove a corrected version of the conjecture put forward in \cite{Alexandrov:2025sig},
which was an attempt to generalize the observation made in \cite{Pioline:2025xgf} 
about a simplification of the modular anomaly equation \eqref{exp-whhr}. 
The simplification does not hold in general, but only when all D4-brane charges are collinear, the case relevant for VW theory.
The basic idea is that in this case the kernel \eqref{refsolRn} greatly simplifies after expressing 
the generalized error functions $\Phi_n^E$ in \eqref{Erefsim}
in terms of their complementary counterparts $\hPhi^M_n$ by means of the relation \eqref{expPhiE-mod}.\footnote{See appendix \ref{ap-generr} 
	for definitions and properties of the generalized error functions.}
The claim is that the sum over Schr\"oder trees disappears and only a single term remains, possibly supplemented
by terms involving Kronecker symbols as in \eqref{defSlin}. The latter terms are non-vanishing only on special charge configurations,
which form low dimensional sublattices of the lattice one sums over in \eqref{Rirf-to-rmRrf}.

It is important to note the Kronecker symbols in \eqref{defSlin} are non-vanishing exactly at the discontinuity loci of the functions $\hPhi^M_n$ 
where, strictly speaking, they were not defined.
Since the corresponding charge configurations do contribute to the theta series, 
if one wants to use $\hPhi^M_n$ as building blocks of the kernel $\scRrf_n$, 
their definition must be extended to the discontinuity loci.
We achieve this in appendix \ref{ap-generr}.
First, we simply require $\hPhi^M_n$ to satisfy the relation \eqref{expPhiE-mod} everywhere, including the problematic loci.
This turns out to be equivalent to defining the complementary generalized error functions by the principal value contour prescription.
Given such an extension, let us introduce the function analogous to \eqref{Erefsim}
\be 
\Mv_n(\bfhgam;\tau_2,\beta)= \hPhi^M_{n-1}\(\{ \bfv_{\ell}\};\sqrt{2\tau_2}\,(\bfq+\beta\bftet ),\sqrt{2\tau_2}\beta\bftet\).
\label{vMdef}
\ee 
Then we have the following
\begingroup
\renewcommand{\thetheorem}{1a}
\begin{theorem}
\label{thm-simple-ref}
In the collinear case where $p_i=r_ip_0$, the kernel $\scRrf_n$ \eqref{refsolRn} simplifies to
\be
\scRrf_n\(\bfhgam;\tau_2,\beta\) =
\sum_{\cI\subseteq \Zv_{n-1}}
\( 
\delta\(\Gamma_{(\Zv_{n-1}\setminus\cI) \perp\cI}\)
\prod_{s=1}^{m_{\cI}} b_{|\cI_s|}\)
\Mv_{|\cI|+1}(\{\hbbgam_s\};\tau_2,\beta),
\label{refsolRn-simple-new}
\ee		
where $\delta(x_\cI)$ was defined in \eqref{def-delta}, $\Gamma_{\ell\perp\cI}$ below \eqref{rel-delperp},
$\cI_s$ are the subsets \eqref{subsets} induced by $\cI$, $m_{\cI}$ is the number of these subsets,
$\hbbgam_s=\sum_{i\in\cI_s}\hgam_i$ are their total charges, and
\be 
b_{n-1}=\frac{2^n(2^n-1)}{n!}\, B_n
\label{defbn}
\ee
are the Taylor series coefficients of $\tanh(x)$, with $B_n$ being the Bernoulli number.
\end{theorem}
\endgroup
\begin{proof}
Since $p_i$ are collinear, one can apply Lemma \ref{lemma1}.
Then following the same reasoning as in \S\ref{subsec-useful}, the $\delta$-factor in \eqref{refsolRn-simple-new}
factorizes into a product of factors depending only on the subsets of charges labeled by $\cI_s$.
Therefore, the statement of the theorem is equivalent to (cf. \eqref{Efrf-tree2})
\be
\scRrf_n =
b_{n}\,\delta(\Gamma_{\Zv_{n-1}})
+\sum_{T\in \IT^{\rm S,\leq 2}_n} 
\Mv_{v_0}
\prod_{v\in V_T\setminus\{v_0\}} 
b_{n_v}\,\delta_{v}.
\label{refsolRn-simple-new2}
\ee	
It is this equality that will be proven below.

To prove this equality, let us apply the relation \eqref{expPhiE-mod} between the generalized error functions to \eqref{Erefsim}.
It is easy to see that, in the case of collinear charges, the result 
can also be written in the form similar to \eqref{refsolRn-simple-new2}, namely, 
\be
\Er_n=
\sgn(\Gamma_{\Zv_{n-1}})
+\sum_{T\in \IT^{\rm S,\leq 2}_n} 
\Mv_{v_0}
\prod_{v\in V_T\setminus\{v_0\}} S^{(0)}_v . 
\label{PhiEPhiM-tree}
\ee
Substituting it into \eqref{refsolRn} and recombining different trees, one obtains 
\be
\scRrf_n = \sum_{T\in\IT_n^{\rm S}}(-1)^{n_T} 
\(S^{(\delta)}_{v_0} \prod_{v\in V_T\setminus\{v_0\} }\Efrf_v
-\Mv_{v_0}
\prod_{v_1\in V_T^{(1)} }S^{(\delta)}_{v_1}
\!\!\!\prod_{v\in V_T \setminus (V_T^{(1)}\cup \{v_0\})}\!\!\!  \Efrf_v\),
\label{refsolRn-2}
\ee
where $V_T^{(1)}$ is the set of vertices of depth 1 and $S^{(\delta)}_v=\Efrf_v-S^{(0)}_v$.
Now it is immediate to see that the statement of the theorem \eqref{refsolRn-simple-new2} follows from the following 
\begin{lemma}
	\be 
	\sum_{T\in\IT_n^{\rm S}}(-1)^{n_T}S^{(\delta)}_{v_0} 
	\prod_{v\in V_T \setminus \{v_0\})}\Efrf_v=b_n\,\delta(\Gamma_{\Zv_{n-1}}).
	\label{propSS}
	\ee
\end{lemma}
\begin{proof}
	From \eqref{Efrf-tree2}, it follows that
	\be
	S^{(\delta)}_n
	=
	\sum_{T\in \IT^{\rm S,\leq 2}_n} 
	e_{n_{v_0}-1}\delta_{v_0} \prod_{v\in V_T\setminus\{v_0\}} S^{(0)}_v,
	\label{Slin-plus-tree}
	\ee
	which shows that this function is non-vanishing only on special charge configurations where at least one
	of the Kronecker symbols is non-zero.
	Substituting this into the l.h.s. of \eqref{propSS} and recombining the trees, one obtains (cf. \eqref{refsolRn-2})
	\be 
	\sum_{T\in\IT_n^{\rm S}}(-1)^{n_T}e_{n_{v_0}-1}\delta_{v_0} \prod_{v_1\in V_T^{(1)} }S^{(\delta)}_{v_1}
	\prod_{v\in V_T \setminus (V_T^{(1)}\cup \{v_0\})}\Efrf_v.
	\ee 
	Now we can again substitute \eqref{Slin-plus-tree} at vertices of depth 1. 
	This replaces the factors $S^{(\delta)}_{v_1}$ by $e_{n_{v_1}-1}\delta_{v_1}$ and 
	$\Efrf_v$ for $v\in V_T^{(2)}$ by $S^{(\delta)}_{v}$.	
	Proceeding in this way, in $d_T$ steps, where $d_T$ is the depth of tree, on arrives at 
	the following expression 
	\be  
	\sum_{T\in\IT_n^{\rm S}}(-1)^{n_T}\prod_{v\in V_T }e_{n_{v}-1}\delta_{v}.
	\ee
	Taking into account that 
	\be 
	\prod_{v\in V_T }\delta_{v}=\delta(\Gamma_{\Zv_{n-1}}),
	\ee 
	the Kronecker symbols factorize and one remains with a sum over Schr\"oder trees with vertices weighted by $e_{n_{v}-1}$.
    According to \eqref{expr-bn}, it is equal to the coefficient $b_n$. As a result, we arrive at the expression on the r.h.s. of \eqref{propSS},
	which proves the Lemma.
\end{proof}
This completes the proof of the theorem.
\end{proof}

Theorem \ref{thm-simple-ref} shows that the kernel $\scRrf_n$ coincides with $\hPhi^M_{n-1}$ up to contributions supported only 
on the discontinuity loci of this function.
Since originally it was not defined there and we had to introduce its extension to these loci anyway,
this suggests that it might be possible to absorb the additional contributions by modifying the extension. 
In fact, the simplest possibility is to take the sum over subsets in \eqref{refsolRn-simple-new} 
as such a new extension. This would absorb all the additional contributions by construction.
However, it appears to be {\it ad hoc} and it is not clear why the new extension should be better than the previous one.

Instead, in appendix \ref{ap-generr} we propose another modification $\hPhi_n^\Msf$, which is distinguished by the property that
it is exponentially suppressed at large $\xbbm$ (and $\bbbm=0$).
It satisfies a modified relation to $\Phi^E_n$ (\eqref{expPhiE-mod3-hEM} vs. \eqref{expPhiE-mod-EhM})
and is related to $\hPhi_n^M$ by \eqref{expPhiE-mod-inv3}.
Although the r.h.s. of this formula looks similar to \eqref{refsolRn-simple-new} upon replacement 
$(\vbbm_i,\xbbm,\bbbm)\mapsto(\bfv_i,\sqrt{2\tau_2}(\bfq+\beta\bftet),\sqrt{2\tau_2}\beta\bftet)$,
they differ by the numerical coefficients which weigh the contribution of each subset.
Nevertheless, remarkably, it turns out that they become equal after a proper symmetrization!
Namely, defining the function $\euM_n$ from $\hPhi^\Msf_{n-1}$ in the same way as in \eqref{vMdef} $\Mv_n$ was defined from $\hPhi^M_{n-1}$, 
we can formulate the following
\begingroup
\renewcommand{\thetheorem}{1b}
\begin{theorem}
\label{thm-simple-ref2}
\addtocounter{theorem}{-1}
For collinear D4-brane charges, one has
\be
\Sym \Bigl\{ (-y)^{\sum_{i<j} \gamma_{ij}}\, \scRrf_n(\bfhgam;\tau_2,\beta)
\Bigr\}=
\Sym \Bigl\{ (-y)^{\sum_{i<j} \gamma_{ij}}\,\euM_n(\bfhgam;\tau_2,\beta)
\Bigr\}.
\label{Rn-simple-new}
\ee
\end{theorem} 
\endgroup
\begin{proof}
Due to \eqref{expPhiE-mod-inv3}, the r.h.s. of \eqref{Rn-simple-new} takes the form
\be 
\Sym \Bigl\{ (-y)^{\sum_{i<j}\gamma_{ij} }
	\!\!\!\sum_{\cI\subseteq \Zv_{n-1}}\!\!\!
	\bsf_{n-|\cI|}(\bfv_{(\Zv_{n-1}\setminus\cI)\perp\cI})\,
	\delta\bigl(\Gamma_{(\Zv_{n-1}\setminus\cI) \perp\cI}\bigr)
	\,	\Mv_{|\cI|+1}(\{\hbbgam_s\};\tau_2,\beta)
	\Bigr\}.
	\label{expr-SymM}
\ee 
According to Lemma \ref{lemma1}, the vectors $\bfv_{\Zv_{n-1}\setminus\cI\perp\cI}$ split into $m_\cI$ mutually orthogonal subsets 
$\{\lambda_s\bfv^{(s)}_\ell\}_{\ell\in\cI'_s}$ \eqref{rel-bfvperp}. 
Due to this, 
\be 
\delta\bigl(\Gamma_{(\Zv_{n-1}\setminus\cI) \perp\cI}\bigr) = 
\prod_{s=1}^{m_\cI}\delta\bigl(\Gamma^{(s)}_{\cI'_s}\bigr),
\label{delta-Thm1p}
\ee 
while the property \eqref{Phi-orth} and the independence of 
$\Phi^E_n$ on the overall scaling of the vectors $\bfv_\ell$ ensure that 
\be
\bsf_{n-|\cI|}(\bfv_{(\Zv_{n-1}\setminus\cI) \perp\cI})
=\prod_{s=1}^{m_\cI}\bsf_{|\cI_s|}\bigl(\{\bfv^{(s)}_\ell\}\bigr).
\ee

Next, we note that $n-1$ conditions $\Gamma_\ell=0$ imply that {\it all} $\gamma_{ij}$ vanish.
This follows from the fact that in the collinear case the lattice \eqref{latp} is $n-1$-dimensional 
and the vectors $\bfv_\ell$ are all linearly independent. Therefore, the vectors $\bfv_{ij}$
can be expressed as their linear combinations.
Due to this observation, each factor in \eqref{delta-Thm1p} can be written as 
\be 
\delta\bigl(\Gamma^{(s)}_{\cI'_s}\bigr)=\prod_{i<j\atop i,j\in \cI_s}\delta_{\gamma_{ij}}.
\label{deltas-Thm1p}
\ee 
As a result, the power of $y$ and 
the function $\Mv_{|\cI|+1}$ depend only on the total charges of the subsets and hence are not affected by 
permutations acting inside each subset. Because of \eqref{deltas-Thm1p}, such permutations also do not affect \eqref{delta-Thm1p}.

Thus, the only factor in \eqref{expr-SymM} which is affected by the permutations inside a subset $\cI_s$ is 
$\bsf_{|\cI_s|}(\{\bfv^{(s)}_\ell\})$. Since these permutations act independently on each subset, this allows to 
replace the coefficients $\bsf_n$ by their symmetrized versions.
But the latter, due to \eqref{bsym}, coincide with the coefficients $b_n$.
After this replacement, the expression \eqref{expr-SymM} becomes identical to the l.h.s. of \eqref{Rn-simple-new}
after using \eqref{refsolRn-simple-new}. This completes the proof.
\end{proof}

The l.h.s. of \eqref{Rn-simple-new} is exactly what appears in the coefficient \eqref{Rirf-to-rmRrf} determining the modular anomaly.
Thus, according to Theorem \ref{thm-simple-ref2}, in the collinear case, 
the kernel of the modular anomaly coincides with the function $\euM_n(\bfhgam;\tau_2,\beta)$,
given by our second extension of the complementary generalized error function, without any additional contributions.

\section{Variation of the modular completion}
\label{sec-varcompl}

Now we turn to VW theory on a Hirzebruch surface and consider the solution presented in \S\ref{subsec-genfunVW}.
In this section we will be interested in the modular completion $\whh_{N,\mu,J}$ \eqref{complFBJ} of the generating series of refined VW invariants.
Similarly to the generating series, the completion carries a dependence on the polarization vector $J$.
For $h_{N,\mu,J}$, the dependence is piece-wise constant and captured by the universal wall-crossing formulae  
\cite{ks,Joyce:2008pc}. In contrast, $\whh_{N,\mu,J}$ are smooth functions of $J$, as follows from \eqref{kerhg-tree}.
Nevertheless, one can ask whether their dependence on $J$ is also governed by some universal structure.
As we will show, the answer to this question appears to be affirmative.

More precisely, similar to the wall-crossing equations, we express the completion of rank $N$ evaluated at polarization $J'$ in terms of
the completions of ranks $N_i\leq N$ evaluated at $J$.
The result is given by the following

\begin{theorem}
\label{th-compl}
\be
\whh_{N,\mu,J'}(\tau,z)-\whh_{N,\mu,J}=\Theta_{N,\mu}\Bigl(\tau,z;\{\wheuF_n(J,J')\},\{\whh_{N',\mu,J}\}\Bigr),
\label{complFBJJ'}
\ee
where $\Theta_{N,\mu}$ is defined in \eqref{def-Theta} and
\be 
\wheuF_n(\bfhgam;J,J')= \sum_{T\in \IT^{\rm S}_n} (-1)^{n_T-1}
\Bigl(\Ev_{v_0}(J')	-\Ev_{v_0}(J)\Bigr)
\prod_{v\in V_T\setminus\{v_0\}}\Ev_v(J).
\label{ansPhin}
\ee
\end{theorem}
\begin{proof}
At the first step, let us consider the expressions \eqref{complFBJ} for the modular completions $\whh_{N,\mu,J}$
as equations on the functions $H^S_{N,\mu}$. It is straightforward to solve them iteratively.
The solution can be represented as a sum over Schr\"oder trees:
\be
H^S_{N,\mu}(\tau,z)=\whTheta_{N,\mu}\Bigl(\tau,z;\{\euF^H_n\},\{\whh_{N',\mu,J}\}\Bigr),
\label{holmod-whh}
\ee
where
\be 
\euF^H_n=\sum_{T\in \IT^{\rm S}_n} (-1)^{n_T}\prod_{v\in V_T}\whFv_v.
\ee 
Indeed, substituting \eqref{holmod-whh} back into \eqref{complFBJ} and using the property \eqref{compTheta}, 
one finds that the resulting kernel vanishes unless $n=1$.

Next, we compute the difference of the completion evaluated at two polarization vectors, $J$ and $J'$.
Using again \eqref{complFBJ}, it can be written as 
\be
\whh_{N,\mu,J'}-\whh_{N,\mu,J}=\Theta_{N,\mu}\Bigl(\{\Delta\whFv_n\},\{H^S_{N',\mu}\}\Bigr),
\label{complFB-diff}
\ee
where $\Delta\whFv_n=\whFv_n(J')-\whFv_n(J)$.
Then we substitute \eqref{holmod-whh} and use the composition property \eqref{compTheta} to write the result 
in the form of \eqref{complFBJJ'}.
The kernel $\wheuF_n(J,J')$ obtained in this way is given by 
\be 
\wheuF_n(J,J')=\sum_{T\in \IT^{\rm S}_n} (-1)^{n_T-1}
\Delta\whFv_{v_0}
\prod_{v\in V_T\setminus\{v_0\}}\whFv_v.
\label{Phin-subs}
\ee 
It remains to show that upon substitution of \eqref{kerhg-tree}, all terms involving factors 
$S^\Bv_v:=\sgn\bigl(-\Bv_{\Ch'(v)}\bigr)$ cancel.
It is easy to see that the remaining terms then reproduce the desired kernel \eqref{ansPhin}.

Substituting \eqref{kerhg-tree} into \eqref{Phin-subs}, one obtains 
\be 
\wheuF_n(J,J')=\sum_{T\in \tilde\IT^{\rm r}_n} 
(\Ev_{v_0}(J')-\Ev_{v_0}(J))
\prod_{v\in V_T^{\rm even}\setminus\{v_0\}}\(-\Ev_v(J)\)\prod_{v\in V_T^{\rm odd}}S^\Bv_v,
\label{Phin-subs2}
\ee 
where $V_T^{\rm even}$ and $V_T^{\rm odd}$ are the sets of vertices (excluding leaves) of even and odd depth, respectively,
and $\tilde\IT^{\rm r}_n\subset\IT^{\rm r}_n$ is the set of rooted trees satisfying the following conditions:
i) the depth is even, ii) the root vertex has more than one child, 
and iii) one cannot have $n_v=n_{v'}=1$ if $v\in V_T^{\rm even}$ and $v'\in \Ch(v)$.
The last two conditions come from the restriction in \eqref{Phin-subs} to Schr\"oder trees.

To see the cancellations, let us a pick up an arbitrary tree $T\in \tilde\IT^{\rm r}_n$ where 
a vertex $v\in V_T^{\rm even}$ has $n_v>1$ and one of its children $v'$ also satisfies $n_{v'}>1$. 
Its contribution gets canceled by the contribution of a similar tree where 
between $v$ and $v'$ there are two additional vertices each having a single child. Indeed, the two trees have the contributions 
differing only by the overall sign. Thus, one remains with contributions of trees where 
for any two consecutive vertices at least one must have only one child.

Since the root cannot have one child, all its children have $n_v=1$. 
Furthermore, such trees with vertices of depth 2 also having only one child have already been used for cancellations
at the very first step. Thus, at depth 2 one has $n_v>1$. But then vertices of depth 3 must have $n_v=1$.
Proceeding in this way, one concludes that all vertices of even depth satisfy $n_v>1$, whereas for all vertices of odd depth $n_v=1$.
Since for such vertices $S^\Bv_v=1$, they can be safely removed and one remains with the sum over Schr\"oder trees 
exactly as in \eqref{ansPhin}.
This completes the proof.
\end{proof}

This theorem shows that the kernel of the theta series that determines the change of the modular completion under a variation of 
the polarization is entirely expressed through the functions $\Ev_n$ identical to the generalized error functions $\Phi^E_{n-1}$.
Hence the modularity of this representation is manifest. More interestingly, the kernel does not depend on the surface data
such as the null vector $\nv(S)$. (Of course, the information on the charge lattice and polarization is still there and cannot be removed.)
This hints that the result \eqref{complFBJJ'} might be valid and have exactly the same form 
in a more general context than VW theory on Hirzebruch surface.
The necessary requirement for such a generalization is that the (mock) modularity of the generating functions continues to be governed 
by the structure presented in \S\ref{sec-anomaly} beyond the large volume attractor chamber,
as it happens in VW theory.

\section{VW invariants at the walls}
\label{sec-VEwall}

Let us now consider the dependence on the polarization vector $J$ of the generating functions $h_{N,\mu,J}$ \eqref{genJhN}
of refined VW invariants.
As it should be, this dependence is piece-wise constant because $J$ enters only through the sign functions and the Kronecker symbols
in the kernel \eqref{kerg}. As soon as one of $\Gamma_\ell(J)$'s vanishes, the kernel and hence the generating function jump.
Thus, the polarization vectors for which there are charges such that at least one $\Gamma_\ell(J)=0$, 
correspond to the walls of marginal stability.\footnote{The only exception is $J\sim c_1(S)$, i.e. the canonical polarization.
	In this case, the jump takes place for two charge configurations that cancel each other.
	This cancellation is possible because the two configurations have the same power of $y$, 
	whereas this is not the case if $J\not\sim c_1(S)$.}
In practice, this means that $J$ should have integer components up to an overall rescaling.

An important point is that the generating functions \eqref{genJhN} are defined for {\it any} $J$.
In particular, this implies that they can be computed at the walls!
An indication that the invariants computed in this way are meaningful was given already in \cite{Alexandrov:2020dyy}:
VW invariants on $\IF_1$ computed at a wall were shown to correctly reproduce VW invariants on $\IP^2$ 
by means of the standard blow-up formulae \cite{Yoshioka:1996,0961.14022,Li:1998nv}.
Therefore, one can expect that they do have a mathematical meaning 
and hence it would be interesting to determine their properties.

One important property of the standard VW invariants is actually lost for the invariants at the walls. 
Recall that the generating functions $h_{N,\mu,J}$ compute rational invariants \eqref{defcref}.
To get the generating functions $\hint_{N,\mu,J}$ of integer invariants, 
one should compute the linear combinations \eqref{def-hint}.
It is easy to check in concrete examples that for $J$ corresponding to a wall, 
the coefficients of $\hint_{N,\mu,J}$ are, in general, {\it not} integers.
Thus, the invariants at the walls cannot be given by a mere counting of some objects.

Next, one can ask whether they are related to the invariants on the two sides of the wall.
It is easy to see that at rank $N=2$, they are simply given by the arithmetic average of these invariants, 
i.e. if $J_\pm$ are small rotations of $J$ to the left/right then
\be 
h_{2,\mu,J}=\hf\(h_{2,\mu,J_+}+h_{2,\mu,J_-}\).
\label{Delh2}
\ee 
But already for $N=3$, this is not true.
The main result of this section is an explicit evaluation of the deviation from the average for arbitrary rank.
Namely, we compute 
\be 
\Delta h_{N,\mu,J}=h_{N,\mu,J}-\hf\(h_{N,\mu,J_+}+h_{N,\mu,J_-}\)
\label{defDh}
\ee 
for $S=\IF_m$.

To present the result, we need to introduce several new quantities:

\begin{itemize}
\item 
First, we define a few objects related to the polarization vector $J^\alpha=(J^1,J^2)$,
which we assume, without loss of generality, to satisfy $\gcd(J^1,J^2)=1$:
\begin{itemize}
	\item 
	$ 
	J_\pm^\alpha=(J^1\pm \eps J^2, J^2\mp\eps J^1),
	\label{Jpm}
	$	
	where $\eps$ is an infinitesimal positive parameter;
	\item 
	$\vrh\in\Lambda_S$, any fixed lattice vector satisfying $J\cdot \vrh=1$;
	\item
	$J_\perp=\veps\cdot J$, where $\veps_{\alpha\beta}=\(\begin{array}{cc}0 & 1\\-1 & 0\end{array}\)$
	and the indices are raised and lowered with help of the	quadratic form $C_{\alpha\beta}$ \eqref{dataFk}.
\end{itemize}
Note the properties\footnote{Here and everywhere below $J^2=J\cdot J$, not the second component of $J$.}
\be 
\veps^T C^{-1} \veps=-C,
\qquad
J_\perp^2=-J^2,
\qquad
J_\perp\cdot J=0.
\label{prop-JJp}
\ee 
	
\item 
Second, we define a theta series analogous to \eqref{Theta-bfN} where the two-dimensional lattice $\Lambda_{\IF_m}$ is replaced by 
the one-dimensional lattice $\IZ$.
More precisely, it can be obtained from \eqref{Theta-bfN} by restricting the sum over charges
to those which have the form $q_i^\alpha=J_\perp^\alpha \qsf_i$. This gives
\be 
\vth^{\perp(\bfN)}_{\lambda,\bflam}(\tau,z;\Fv_n)=
\sum_{\qsf_i\in N_i\IZ+\lambda_i\atop \sum_{i=1}^n \qsf_i=\lambda}
\Fv_n(\bfhgam^\perp)\,
\q^{-\frac{J^2}{2}\,Q_n(\bfhgam^\perp)}\, y^{(c_1\cdot J_\perp)\sum\limits_{i<j}\gamma^\perp_{ij}},
\label{def-vthperp}
\ee 
where $\lambda\in\IZ_N$, $\lambda_i\in\IZ_{N_i}$ are residue classes, 
the quadratic form $Q_n$ is given by the same equation \eqref{defQlr-VW} as before but now evaluated on 
the set of $\hgam^\perp_i=(N_i,\qsf_i)$, and 
\be 
\gamma^\perp_{ij}=N_i \qsf_j -N_j \qsf_i.
\ee 
As in \eqref{defGammaVW}, we also define 
\be
\Gamma^\perp_\ell:=\sum_{i=1}^\ell\sum_{j=\ell+1}^n\gamma^\perp_{ij}.
\label{defGamma-perp}
\ee

\item
Finally, analogous to \eqref{def-Theta}, we define 
\be 
\hspace{-0.5cm}
\Theta^J_{N,\mu}\Bigl(\{\cF_n\},\{H_{N',\mu}\}\Bigr)=
\sum_{n=3\atop n{\rm -odd}}^N \frac{1}{2^{n-1}}\sum_{\sum_{i=1}^n N_i=N}\sum_{\bfmu}
\(\prod_{i=1}^n\delta^{(1)}_{J\cdot \bigl(\frac{\mu_i}{N_i}-\frac{\mu}{N}\bigr)}\)
\vth^{\perp(\bfN)}_{\vrh\cdot\veps\cdot\mu,\vrh\cdot\veps\cdot\bfmu}(\cF_n)
\prod_{i=1}^n H_{N_i,\mu_i},
\label{def-Theta-perp}
\ee 
where we dropped the arguments $(\tau,z)$ to avoid cluttering.
The two main differences with respect to \eqref{def-Theta}
are the presence of the Kronecker symbols imposing restrictions on the residue classes $\mu_i$ and the range of summation for $n$.
Note also that the ambiguity in the vector $\vrh$ does not affect the construction because it corresponds to the shift of 
the vector indices of the theta series \eqref{def-vthperp} that does not change its summation range.

\end{itemize}

Using these quantities, we can now formulate the following
\begin{theorem}
\label{th-wall}
If $J$ corresponds to a wall of marginal stability, then 
\be
\Delta h_{N,\mu,J}
=\Theta^J_{N,\mu}\Bigl(\{\Fp_n\},\{h_{N',\mu,J}\}\Bigr),
\label{DelhNresh}
\ee 
where the kernel is given by
\be 
\Fp_n(\bfhgam^\perp)=-\sum_{m=1}^{n}\sum_{\sum_{s=1}^{m}n_s=n}S(\Gamma^\perp_{\{j_s\}})
\prod_{s=1}^{m} b_{n_s}
\label{kerperp}
\ee 
the coefficients $b_n$ are defined in \eqref{defbn}, the function $S(\Gamma_{\cI})$ in \eqref{defSlin}, and 
$j_s=\sum_{t=1}^s n_s$, $s=1,\dots,m-1$.
\end{theorem}
\begin{proof}
From \eqref{genJhN}, it follows that the difference \eqref{defDh} is given by 
\be
\Delta h_{N,\mu,J}
=\Theta_{N,\mu}\Bigl(\{\Delta\Fv_n(J)\},\{H^{\IF_m}_{N',\mu}\}\Bigr),
\label{DelhN-int}
\ee 	
where 
\be
\Delta\Fv_n(J)=\Fv_n(J)-\hf\(\Fv_n(J_+)+\Fv_n(J_-)\).
\ee
Substituting \eqref{kerg} and taking into account that $\Gamma_\ell(J_\pm)\ne0$ and have the same sign if $\Gamma_\ell(J)\ne 0$
and $\sgn(\Gamma_\ell(J_+))=-\sgn(\Gamma_\ell(J_-))$ otherwise, 
one finds 
\be 
\Delta\Fv_n(J)=
\sum_{\cI\subseteq\Zv_{n-1}\atop |\cI|\in 2\IN}
\delta(\Gamma_\cI(J))\,\Bigl(e_{|\cI|}-S(\Gamma_\cI(J_+))\Bigr)
\prod_{\ell\in \Zv_{n-1}\setminus \cI}\Bigl(\sgn(\Gamma_{\ell\perp\cI}(J))-\sgn(\Bv_\ell)\Bigr).
\label{kerDelF}
\ee 
As in \S\ref{subsec-useful}, Lemma \ref{lemma1} ensures that, for each subset $\cI$, 
the second factor can be represented as a product of $m_\cI$ factors each 
depending only on the subset of charges $\bfhgam^{(s)}=\{\hgam_i\}_{i\in\cI_s}$,
while the first factor obviously depends only on the total charges $\hbbgam_s$ of these subsets.
Then applying the bijection \eqref{ident-sumset} between the sum over subsets and the sum over ordered partitions,
one can rewrite the kernel \eqref{kerDelF} as
\be 
\Delta\Fv_n(J)=\sum_{m=3}^n \sum_{\sum_{s=1}^{m} n_s=n}\Fv^{\perp(0)}_m(\{\hbbgam_s\};J)
\prod_{s=1}^m\Fv^{(0)}_{n_s}(\bfhgam^{(s)};J),
\label{kerDelF2}
\ee 
where
\bea 
\Fv^{\perp(0)}_n(\bfhgam;J)&=& \delta_n(\bfhgam;J)\,\Bigl(e_{n-1}-S(\Gamma_{\Zv_{n-1}}(J_+))\Bigr),
\label{kerperpz}
\\
\Fv^{(0)}_n(\bfhgam;J)
&=&
\prod_{\ell=1}^{n-1}\Bigl(\sgn(\Gamma_\ell(J))-\sgn(\Bv_\ell)\Bigr).
\label{kerg0}
\eea
and we introduced a shorthand notation
\be 
\delta_n(\bfhgam;J)=\delta^{(2)}_{n-1}\,\delta(\Gamma_{\Zv_{n-1}}(J)).
\label{deltn}
\ee 
Here the first Kronecker symbol ensures that $|\cI|$ is even in \eqref{kerDelF}.
The expression \eqref{kerDelF2} has exactly the same form as \eqref{compos-ker}. 
Therefore, we can apply the composition property \eqref{compTheta} in the opposite direction to get 
\be
\Delta h_{N,\mu,J}
=\Theta_{N,\mu}\Bigl(\{\Fv^{\perp(0)}_n(J)\},\{h^{(0)}_{N',\mu,J}\}\Bigr),
\label{DelhNresh0}
\ee 
where we introduced  
\be 
h^{(0)}_{N,\mu,J}=\Theta_{N,\mu}\Bigl(\{\Fv^{(0)}_n(J)\},\{H^S_{N',\mu}\}\Bigr).
\label{genJhN0}
\ee

The result \eqref{DelhNresh0} is already very close to the statement \eqref{DelhNresh} we want to prove.
The main difference between them is that the former is written in terms of the functions $h^{(0)}_{N,\mu,J}$, 
whereas the latter in terms of our standard generating functions $h_{N,\mu,J}$. 
In turn, these two differ by the presence or absence of the terms
involving Kronecker symbols in their respective kernels, \eqref{kerg} and \eqref{kerg0}.
Note that the missing terms are precisely the ones that become relevant at the walls
and were collected in the kernel \eqref{kerperpz}.
Therefore, to reproduce \eqref{DelhNresh}, we have to go one step back and reshuffle the expression \eqref{kerDelF2}
by moving some of the $\delta$-terms from $\Fv^{\perp(0)}_n$ to $\Fv^{(0)}_n$ before applying the composition property.
The required reshuffling is provided by the following 

\begin{lemma}
\label{lemma-ker}
\be 
\Delta\Fv_n(J)=\sum_{m=3}^n \sum_{\sum_{s=1}^{m} n_s=n}\Fv^\perp_m(\{\hbbgam_s\};J)
\prod_{s=1}^m\Fv_{n_s}(\bfhgam^{(s)};J),
\label{kerDelF3}
\ee 
where
\be 
\Fv^\perp_n(\bfhgam;J)= -\delta_n(\bfhgam;J)
\sum_{m=1}^{n}\sum_{\sum_{s=1}^{m}n_s=n}S(\Gamma_{\{j_s\}}(J_+))
\prod_{s=1}^{m} b_{n_s}.
\label{kerperpz2}
\ee 
\end{lemma}
\begin{proof}
We prove the statement by bringing \eqref{kerDelF3} to \eqref{kerDelF2}.
As the first step, note that, due to Lemma \ref{lemma1} and the bijection \eqref{ident-sumset}, 
the kernel $\Fv_n$ \eqref{kerg} can be rewritten as 
\be
\Fv_n(\bfhgam;J)=\sum_{m=1}^{n}\sum_{\sum_{s=1}^{m}n_s=n} e_{m-1} \delta(\Gamma_{\{j_s\}}(J))
\prod_{s=1}^m\Fv^{(0)}_{n_s}(\bfhgam^{(s)};J).
\label{kerg2}
\ee
Substituting this into \eqref{kerDelF3} leads to the representation of $\Delta\Fv_n(J)$
of the form \eqref{kerDelF2} where the kernel $\Fv^{\perp(0)}_n$ is replaced by\footnote{The Kronecker symbols 
	in $\Fv^\perp_m$ can, of course, be omitted in this formula because they are already taken into account by 
	the ones in front of the sum. The same comment applies to \eqref{kerperp2} and \eqref{deriv-kerFpz}.}
\be 
\delta_n(\bfhgam;J) \sum_{m=3}^n \sum_{\sum_{s=1}^{m} n_s=n}\Fv^\perp_m(\{\hbbgam_s\};J)
\prod_{s=1}^m e_{n_s-1},
\label{newker}
\ee 
where we took into account that $e_n$ is non-vanishing only for even $n$.
Next, we note that, applying the recursive relation \eqref{b-eb} to the $m=1$ term in \eqref{kerperpz2}, 
this kernel is nicely expressed through \eqref{kerperpz} as 
\be 
\Fv^\perp_n(\bfhgam;J)=\delta_n(\bfhgam;J)
\sum_{m=3}^{n}\sum_{\sum_{s=1}^{m}n_s=n}\Fv^{\perp(0)}_m(\{\hbbgam_s\};J)
\prod_{s=1}^{m} b_{n_s}.
\label{kerperp2}
\ee 
Substituting this into \eqref{newker}, one finds 
\be 
\delta_n(\bfhgam;J)
\sum_{m=3}^{n}\sum_{\sum_{s=1}^{m}n_s=n}\Fv^{\perp(0)}_m(\{\hbbgam_s\};J)
\prod_{s=1}^{m} c_{n_s},
\label{deriv-kerFpz}
\ee
where 
\be 
c_n=\sum_{m=1}^{n}b_m\sum_{\sum_{s=1}^{m}n_s=n}
\prod_{s=1}^{m} e_{n_s-1}=\delta_{n-1}
\ee 
and the second equality follows from the identity \eqref{ident-be}.
Hence, the kernel \eqref{deriv-kerFpz} reduces to $\Fv^{\perp(0)}_n(J)$.
This reproduces \eqref{kerDelF2} and completes the proof.
\end{proof}

As above, the result \eqref{kerDelF3} implies that 
\be
\Delta h_{N,\mu,J}
=\Theta_{N,\mu}\Bigl(\{\Fv^\perp_n(J)\},\{h_{N',\mu,J}\}\Bigr),
\label{DelhNresh2}
\ee 
where we set $\Fv^\perp_1(J)=0$.
Thus, to prove the theorem, it remains to analyze the theta series $\vth^{(\bfN)}_{\mu,\bfmu}(\Fv^\perp_n(J))$.
More precisely, we need to understand the effect of the $\delta$-factor \eqref{deltn} in the kernel \eqref{kerperpz2}.
It imposes the conditions $\bfv_\ell\ast\bfq=0$ for all $\ell=1,\dots, n-1$.
From \eqref{bfvVW} and the fact that $\bffv_\ell\star\boldsymbol{1}=0$, where $\boldsymbol{1}$ is the unit vector,
it follows that these equations are equivalent to the condition $J\cdot \bfq \sim \boldsymbol{1}$.
Taking into account that $\sum_{i=1}^n q_i=\mu$, one therefore obtains the constraints
\be 
\frac{J\cdot q_i}{N_i}=\frac{J\cdot \mu}{N}\, .
\label{eq-qJmu}
\ee 
They have a solution if and only if the residue classes satisfy 
\be 
J\cdot \(\frac{\mu_i}{N_i}-\frac{\mu}{N}\)\in\IZ,
\label{cond-mumu}
\ee 
which reproduces the product of Kronecker symbols in \eqref{def-Theta-perp}.
The full solution to \eqref{eq-qJmu} is found in appendix \ref{ap-lat}
and can be written explicitly using $J_\perp$, $\vrh$ and $\veps$ introduced above \eqref{prop-JJp}.
It is given in \eqref{sol-G0} and can be presented as
\be 
q_i=J_\perp\qsf_i+\frac{N_i}{N}\, (J\cdot\mu)\vrh,
\qquad 
\qsf_i\in N_i\IZ+\vrh\cdot\veps\cdot\mu_i,
\qquad 
\sum_{i=1}^n \qsf_i=\vrh\cdot\veps\cdot\mu.
\label{sol-G0main}
\ee 

Implementing these constraints in the theta series \eqref{Theta-bfN}, it is easy to see that the last term in \eqref{sol-G0main}
cancels everywhere. 
Furthermore, for the charges \eqref{sol-G0main} one has 
\be 
\Gamma_{\ell}(J_+)=\eps|J|^2\Gamma^\perp_\ell,
\ee 
where $\Gamma^\perp_\ell$ was defined in \eqref{defGamma-perp}, 
we denoted $|J|^2=(J^1)^2+(J^2)^2$ and used $J_+\cdot J_\perp=\eps|J|^2$.
As a result, the kernel $\Fv^\perp_n(J)$ \eqref{kerperpz2}, up to $\delta$-factors, reproduces the kernel $\Fp_n$ \eqref{kerperp},
and we arrive at the following relation
\be 
\vth^{(\bfN)}_{\mu,\bfmu}(\Fv^\perp_n(J))=\delta^{(2)}_{n-1}\,
\(\prod_{i=1}^n\delta^{(1)}_{J\cdot \bigl(\frac{\mu_i}{N_i}-\frac{\mu}{N}\bigr)}\)
\vth^{\perp(\bfN)}_{\vrh\cdot\veps\cdot\mu,\vrh\cdot\veps\cdot\bfmu}(\Fp_n),
\ee 
where $\vth^{\perp(\bfN)}_{\lambda,\bflam}$ was defined in \eqref{def-vthperp}.
Substituting this relation into \eqref{DelhNresh2}, it is immediate to see that $\Theta_{N,\mu}$ reduces to 
$\Theta^J_{N,\mu}$ defined in \eqref{def-Theta-perp} and evaluated on $\{\Fp_n\}$ and $\{h_{N',\mu,J}\}$.
This proves \eqref{DelhNresh} and hence the theorem.
\end{proof}

Theorem \ref{th-wall} represents the deviation from the arithmetic average, $\Delta h_{N,\mu,J}$, 
in terms of the generating functions $h_{N_i,\mu,J}$ with $N_i<N$.
Thus, it can be seen as an iterative relation between the generating functions at the walls.
An important observation is that it appears to be independent of the surface data in the same sense as the difference of 
the completions found in \S\ref{sec-varcompl}.
More precisely, to see this independence, one should properly generalize 
the theta series $\vth^{\perp(\bfN)}_{\lambda,\bflam}$ \eqref{def-vthperp}. 
Of course, the fact that it is built from one-dimensional 
lattices is a direct consequence of the fact that $\Lambda_{\IF_m}$ is two-dimensional.
But it is straightforward to extend it to more general cases.
For example, for del Pezzo surfaces, the result for $\Delta h_{N,\mu,J}$ is expected to have the same form \eqref{DelhNresh},
but with $\vth^{\perp(\bfN)}_{\lambda,\bflam}$ such that the lattice $\IZ$ in its definition 
is replaced by the $(b_2(S)-1)$-dimensional sublattice of $\Lambda_S$ orthogonal to $J$.

It is tempting to suggest that the formula \eqref{DelhNresh} might be an example of some universal equations expressing 
generating functions of topological invariants computed on the walls of marginal stability.\footnote{In this respect, 
	it is interesting to note that the kernel \eqref{kerperp} of the theta series coupling different generating functions 
	has a structure very similar to the kernel \eqref{refsolRn-simple-new} appearing in Theorem \ref{thm-simple-ref}.}
If such `wall-sitting' equations exist, it would be interesting to find them for other types of invariants 
whose definition can be extended to the walls.

\subsection{Example: $N=3$}
\label{subsec-n3}

Let us make the result given by Theorem \ref{th-wall} explicit in the simplest non-trivial case of rank $N=3$.
In this case there is an additional simplification 
\be
\Delta h_{3,\mu,J}=\Delta \hint_{3,\mu,J}.
\label{Delh3int}
\ee 
It holds because the term that is needed to convert to integer invariants is proportional to $h_1$ (see \eqref{def-hint})
and does not experience any wall-crossing.\footnote{This statement and hence the relation similar to \eqref{Delh3int} 
	hold for any rank $N$ given by a prime number.
And due to \eqref{Delh2}, it is also true for $N=4$.} 
Another simplification is that $\Theta^J_{3,\mu}$ \eqref{def-Theta-perp} consists of a single term with $n=3$ and hence all $N_i=1$.
As a result, we get
\be 
\Delta \hint_{3,\mu,J}=\frac14\,\delta^{(3)}_{J\cdot \mu}\, h_1^3 
\sum_{\qsf_i\in \IZ\atop \sum_{i=1}^n \qsf_i=\vrh\cdot\veps\cdot\mu}
\(\frac13\(1-\delta_{\Gamma^\perp_1}\delta_{\Gamma^\perp_2}\)-\sgn(\Gamma^\perp_1) \sgn(\Gamma^\perp_2))\)
\q^{-\frac{J^2}{2}\,Q_3(\bfhgam^\perp)}\, y^{(c_1\cdot J_\perp)\sum\limits_{i<j}\gamma^\perp_{ij}},
\label{Delhint3res1}
\ee 
where $h_1=\I/(\eta(\tau)\theta_1(\tau,2z))$.
Parametrizing 
\be 
\qsf_1=\frac13\,\lambda-k_1,
\qquad
\qsf_2=\frac13\,\lambda+k_1-k_2,
\qquad
\qsf_3=\frac13\,\lambda+k_2,
\ee 
where $\lambda=\vrh\cdot\veps\cdot\mu$,
various entries of the theta series simplify as
\be 
\Gamma^\perp_\ell=3k_\ell, 
\qquad
Q_3(\bfhgam^\perp)=-6(k_1^2+k_2^2-k_1 k_2),
\qquad 
\sum\limits_{i<j}\gamma^\perp_{ij}=2(k_1+k_2).
\ee
As a result, one can rewrite \eqref{Delhint3res1} as 
\be
\Delta \hint_{3,\mu,J}=
\frac14\,\delta^{(3)}_{J\cdot \mu}\,h_1^3\, \vartheta_{2,\vrh\cdot\veps\cdot\mu}\bigl(J^2\tau,2(c_1\cdot \veps\cdot J)z \bigr),
\label{Delh3res}
\ee
where we defined the following false theta series with the quadratic form given by the one of the $A_2$ lattice
\be 
\vartheta_{2,\lambda}(\tau,z)=
\sum_{k_\ell\in\IZ+\frac{\ell\lambda}{3}} 
\(\frac{1}{3}\(1-\delta_{k_1}\delta_{k_2}\)-\sgn(k_1)\sgn(k_2)\)
\q^{k_1^2+k_2^2-k_1 k_2}
y^{k_1+k_2}.
\label{def-vthperp3A2}
\ee

\acknowledgments
SA thanks Hülya Argüz, Pierrick Bousseau and Dominic Joyce for the warm hospitality and interesting discussions during his visit  
to the University of Oxford, which gave rise to this project.
We are also grateful to Jan Manschot, Caner Nazaroglu and Boris Pioline for valuable correspondence. 
The work of AC was funded by Ramanujan fellowship RJF/2023/000070, ANRF, India.

\appendix

\section{Generalized error functions}
\label{ap-generr}

In this appendix we review the definition and properties of the generalized error functions, which
are used to construct kernels of indefinite theta series and their modular completions. 
There are two types of such functions, $E_n$ and $M_n$, as well as their `boosted' versions $\Phi^E_n$ and $\Phi^M_n$.
The complementary generalized error functions $M_n$ (and $\Phi^M_n$) are discontinuous.
The main new result of the appendix is an extension of the definition of $M_n$ to the loci of discontinuity.
We propose two such extensions and prove several properties which they satisfy.

\subsection{Definition and properties}

The generalized error functions $E_n$ and their complementary counterpart $M_n$ have been introduced in 
\cite{Alexandrov:2016enp,Nazaroglu:2016lmr} (see also \cite{kudla2016theta,Funke_2019}).
They are defined by the following integral representations
\bea
E_n(\cM;\vu)&=& \int_{\IR^n} \de \vu' \, e^{-\pi(\vu-\vu')^2} \prod_{i=1}^n \sign(\cM^T \vu')_i\, ,
\label{generr-E}
\\
M_n(\cM;\vu)&=&\(\frac{\I}{\pi}\)^n |\det\cM|^{-1} \int\limits_{\IR^n-\I \vu}\de\vu'\,
\frac{e^{-\pi \vu'^2 -2\pi \I \vu'\cdot \vu}}{\prod_{i=1}^n(\cM^{-1}\vu')_i}\, ,
\label{modMn}
\eea
where $\vu=(u_1,\dots,u_n)$ is $n$-dimensional vector, $\cM$ is $n\times n$ invertible
matrix of parameters, and we use the standard Euclidean metric for the scalar product $\vu'\cdot \vu=\sum_{i=1}^n u'_i u_i$.
The two sets of functions generalize the ordinary error and complementary error functions because at $n=1$ 
one has $E_1(u)=\Erf(\sqrt{\pi}\, u)$ and $M_1(u)=-\sgn(u)\Erfc(\sqrt{\pi}\, u)$.
In this trivial case both functions do not depend on the argument $\cM$, 
and for generic $n$ the information carried by $\cM$ is also highly redundant. 
For example, for $n=2$ and 3, the functions depend only on 1 and 3 parameters, respectively.
We refer to \cite{Nazaroglu:2016lmr} for precise symmetry properties.

To present other important properties of the generalized error functions, we need to introduce a few notations.
Let $\vm_i=\cM_{\cdot\,i}$ be the $n$-dimensional vector representing the $i$-th column of matrix $\cM$.
We denote by $\tvm_i$ the set of dual vectors, i.e. the vectors satisfying $\tvm_i\cdot \vm_j=\delta_{ij}$. 
It is clear that $\cM^{-1}_{ij}=(\tvm_i)_j$.

Furthermore, given a subset $\cI\subseteq\Zv_n$, we define $V_\cI=\Span\{\vm_i\}_{i\in\cI}$
and a projector $\cB_\cI$ on this subspace. 
As in the main text, $\vm_{i\perp\cI}$ will denote the projection of $\vm_i$ orthogonal to $V_\cI$.
But we will also use $\tvm_{i\perp\cI}$, the projection of $\tvm_i$ orthogonal to
$\tV_\cI=\Span\{\tvm_i\}_{i\in\cI}$, and the projector $\tcB_\cI$ on $\tV_\cI$.
It is useful to note that the projection orthogonal to $\tV_{\Zv_n\setminus\cI}$ is equivalent to the projection on $V_\cI$.
Therefore, $\{\tvm_{i\perp\Zv_n\setminus\cI}\}_{i\in\cI}$ 
has the meaning of the set of vectors dual to the subset $\{\vm_i\}_{i\in\cI}$.
Finally, we define two $|\cI|\times|\cI|$ matrices
\be 
(\cM_\cI)_{ij}=\sum_{k=1}^n\cB_{\cI,ik}\cM_{kj},
\qquad
(\tcM_\cI)_{ij}=\sum_{k=1}^n\tcB_{\cI,ik}\cM_{kj},
\qquad
i,j\in\cI,
\label{subsetvar}
\ee
and two $|\cI|$-component vectors
\be  
\vu_\cI=\cB_\cI\vu,
\qquad
\tvu_\cI=\tcB_\cI\vu.
\label{subsetvar-vec}
\ee 
The matrix $\tcM_\cI$ can also be seen as the set of vectors $\{\vm_{i\perp\Zv_n\setminus\cI}\}_{i\in\cI}$,
while $\tvu_\cI$ is the projection of $\vu$ on their span.
It is obvious that $(\cM_\cI)_\cJ=\cM_\cJ$ for $\cJ\subset\cI$.

In terms of these notations, the generalized error functions satisfy the following properties \cite{Nazaroglu:2016lmr}:
\begin{itemize}
\item 
The functions $E_n(\cM;\vu)$ are smooth functions of $\vu$.	
However, their limit at large $\vu$ exhibit discontinuities. 
The precise property reads as \cite{Alexandrov:2024jnu} 
\be
\lim_{\lambda\to\infty}E_n(\cM;\lambda\vu)=\Ssf_{n}(\cM;\vu):=
\sum_{\cI\subseteq \Zv_{n}}
\esf_{|\cI|}(\cM_\cI)\,\delta_{\vu_\cI}
\prod_{j\in \Zv_{n}\setminus \cI} \sgn (\vm_j\cdot\vu),
\label{large-x-PhiE-EM}
\ee
where $\esf_n(\cM)=E_n(\cM;0)$.

\item 
In contrast, the functions $M_n(\cM;\vu)$ exhibit discontinuities where $\tvm_i\cdot \vu=0$.
At these loci, the poles of the integrand in \eqref{modMn} happen to be exactly at the integration contour,
and without an additional prescription the functions are not defined.
Away from these loci, $M_n(\cM;\vu)$ are exponentially suppressed at large values of $\vu$.

\item 
Both sets of functions are even/odd for even/odd $n$:
\be 
E_n(\cM;-\vu)=(-1)^n E_n(\cM;\vu),
\qquad
M_n(\cM;-\vu)=(-1)^n M_n(\cM;\vu).
\ee 
In particular, this implies that $e_n(\cM)=0$ for odd $n$. In contrast, $\cM_n$ is not defined at $\vu=0$.

\item 
If the set of $\vm_i$ can be split into two mutually orthogonal subsets, $\{\vm_i\}_{i\in\cI}$
and $\{\vm_j\}_{i\in\cJ}$ with $\cI\cup\cJ=\Zv_n$, then the generalized error functions factorize as
\be 
E_n(\cM;\vu)=E_{|\cI|}(\cM_\cI;\vu_\cI)\, E_{|\cJ|}(\cM_\cJ;\vu_\cJ)
\label{FactorEM}
\ee 
and similarly for $M_n$.

\item	
The two sets of generalized error functions can be expressed through each other by means of the following
equivalent relations 
\bea
E_n(\cM;\vu)&=&\sum_{\cI\subseteq \Zv_n}M_{|\cI|}(\cM_\cI;\vu_\cI)
	\prod_{j\in \Zv_n\setminus \cI}\sign(\vm_{j\perp\cI}\cdot\vu),
\label{expPhiE-mod-EM}
\\
M_n(\cM;\vu)&=& \sum_{\cI\subseteq \Zv_n}(-1)^{n-|\cI|}E_{|\cI|}(\cM_\cI;\vu_\cI)
	\prod_{j\in \Zv_n\setminus \cI}\sign(\tvm_j\cdot\vu) .
\label{expPhiE-mod-inv-EM}
\eea
The equivalence of the relations follows from the sign identity
\be
\sum_{\cI\,:\,\cL\subseteq\cI\subseteq \Zv_n}(-1)^{|\cI|-|\cL|}
\prod_{i\in \cI\setminus \cL}\sign(\tvm_{i\perp \Zv_n\setminus\cI}\cdot\vu)
\prod_{j\in \Zv_n\setminus \cI}\sign(\vm_{j\perp\cI}\cdot\vu)
=0,
\label{signident-EM}
\ee
where $\cL\subset\Zv_n$ is a fixed subset.
\end{itemize}

\subsection{Extension to the loci of discontinuity}
\label{ap-extension}

Although the complementary generalized error functions are not defined at the loci of their discontinuity,
one can ask whether their definition can be extended to these loci.
One way to do this is to provide a contour prescription for the integral representation \eqref{modMn}.
However, there is another natural way. Namely, we can {\it declare} that the relations \eqref{expPhiE-mod-EM} and 
\eqref{expPhiE-mod-inv-EM} continue to hold at the discontinuity loci 
(with the standard prescription for the sign function that $\sgn(0)=0$).\footnote{One can ask whether the two relations remain still equivalent.
This is indeed the case because the sign identity \eqref{signident-EM} turns out to hold for all $\vu$, as can be seen from 
its proof in \cite[\S A]{Nazaroglu:2016lmr}.}
Since $E_n$ are smooth and defined everywhere, this provides also an unambiguous definition for $M_n$. 
Then following the proof of \eqref{expPhiE-mod-inv-EM} in \cite[Prop.3.9]{Nazaroglu:2016lmr} backwards, it is easy to see that 
such a definition corresponds to the {\it principal value} prescription for the contour integral in \eqref{modMn}.
We will continue to represent the resulting functions by $M_n$.

However, note that, as the main text demonstrates, various constructions involve products of sign functions
supplemented by terms where some of the sign functions are replaced by Kronecker symbols weighted by certain coefficients
(see, e.g., \eqref{defSlin} or \eqref{kerg}).
This suggests that we can also allow for such additional terms in the relations between the generalized error functions.
They would affect the definition of $M_n$ only at the discontinuity loci.

Of course, this possibility gives rise to a huge ambiguity in the choice of terms to add.
It is natural to restrict to those structures that already appear in our context.
For instance, one possibility is to replace the sign product in \eqref{expPhiE-mod-EM} by 
$S(\{\vm_{j\perp\cI}\cdot\vu\}_{j\in \Zv_n\setminus \cI})$
where the function $S(\{x_i\})$ is defined in \eqref{defSlin}.
However, we find that the best choice is to use in this replacement, instead of $S(\{x_i\})$, the function $\Ssf_n(\cM;\vu)$
appearing in the large $\vu$ limit of $E_n$ \eqref{large-x-PhiE-EM}.
Thus, we suggest to introduce {\it new} functions $\Msf_n$ that satisfy 
\be
E_n(\cM;\vu)=\sum_{\cI\subseteq \Zv_n}\Msf_{|\cI|}(\cM_\cI;\vu_\cI)\,
\Ssf_{n-|\cI|}(\tcM_{\Zv_n\setminus\cI};\tvu_{\Zv_n\setminus\cI}).
\label{expPhiE-mod3-EM}
\ee
Comparing \eqref{expPhiE-mod3-EM} to \eqref{expPhiE-mod-EM}, one can read off the relation between 
the two extensions of the complementary generalized error functions
\be 
M_n(\cM;\vu)=\sum_{\cI\subseteq \Zv_n}
\esf_{|\Zv_n\setminus\cI|}(\tcM_{\Zv_n\setminus\cI})\,\delta_{\tvu_{\Zv_n\setminus\cI}}
\,\Msf_{|\cI|}(\cM_\cI;\vu_\cI),
\label{rel-hPhi-tPhi}
\ee
which explicitly shows that they differ only at the discontinuity loci.
Unfortunately, we were not able to find a simple contour prescription in the integral representation \eqref{modMn} which 
would reproduce the extension described by $\Msf_n$.

To justify the new definition \eqref{expPhiE-mod3-EM}, let us recall that $M_n$ are exponentially suppressed at large $\vu$.
While this property was proven only away from their discontinuity loci, it continues to hold also
at these loci for the extension denoted by $M_n$.
Indeed, applying the same limit as in \eqref{large-x-PhiE-EM} to the relation \eqref{expPhiE-mod-EM}, one finds 
that it is satisfied only if 
\be
\lim_{\lambda\to\infty}M_n(\cM;\lambda\vu)=
\esf_{n}(\cM)\,\delta_{\vu}.
\label{large-u-M}
\ee
Thus, for large $\vu$ the function does vanish, but it is non-vanishing at the origin for even $n$.
While this is not a problem for $M_n$, the functions that are really relevant for constructions of indefinite theta series are 
the boosted generalized error functions defined below in \S\ref{subsec-boost}.
Their second argument $\xbbm$ lives in $\IR^d$ carrying a quadratic form such that $\IR^n$ where $\vu$ lives is 
a positive definite subspace of $\IR^d$ and $\vu$ can be seen as a projection of $\xbbm$ to this subspace. 
As a result, for the boosted function $\Phi^M_n$ to be exponentially suppressed at large $\xbbm$,
not only $M_n$ should satisfy the same property for large $\vu$, but it also should vanish at $\vu=0$.
The definition \eqref{expPhiE-mod3-EM} precisely ensures this property.

For our purposes, it is important to find the inverse relations to \eqref{expPhiE-mod3-EM} and \eqref{rel-hPhi-tPhi}.
They are given by the following
\begin{proposition}\label{prop-MEM}
\begin{subequations}
\bea
\hspace{-1.4cm}
\Msf_n(\cM;\vu)&=&
\sum_{\cI\subseteq \Zv_{n}}
\bsf_{n-|\cI|+1}(\tcM_{\Zv_n\setminus\cI})\,\delta_{\tvu_{\Zv_n\setminus\cI}}\,
M_{|\cI|}(\cM_\cI;\vu_\cI)
\label{expPhiE-mod-inv3-EM}
\\
&=& \sum_{\cI\subseteq \Zv_{n}}E_{|\cI|}(\cM_\cI;\vu_\cI)
\!\!\sum_{\cJ\subseteq \Zv_n\setminus\cI}\!\!
(-1)^{n-|\cI|-|\cJ|}\bsf_{|\cJ|+1}(\tcM_\cJ)\,\delta_{\tvu_\cJ}
\!\!\!\!\!\prod_{j\in \Zv_n\setminus (\cI\cup\cJ)}\!\!\!\!\!
\sign(\tvm_j\cdot\vu),
\label{inv-tPhi-PhiE-gen-EM}
\eea
\end{subequations}
where the coefficients $\bsf_n(\cM)$ are given by\footnote{We added the shift $+1$ in the index 
	because these coefficients turn out to be closely related to the coefficients $b_n$ discussed in appendix \ref{ap-coef}.
	The relation is given by \eqref{bsym}.} 
\be  
\bsf_{n+1}(\cM)=\sum_{m=1}^n(-1)^m \sum_{\cup_{s=1}^m \cI_s =\Zv_n}\prod_{s=1}^m 
\esf_{|\cI_s|}\Bigl(\cM_{\cI_s\perp \cup_{\ell=1}^{s-1}\cI_\ell}\Bigr),
\label{def-bsf-EM}
\ee 
where the sum goes over ordered decompositions of $\Zv_n$ into non-empty disjoint subsets
and by $\cM_{\cI\perp\cJ}$ we mean the $|\cI|\times|\cI|$ matrix constructed from $\{\vu_{i\perp\cJ}\}_{i\in\cI}$.
\end{proposition} 	
\begin{proof}
Taking into account that $\delta_{\tvu_\cJ}=\prod_{i\in\cJ}\delta_{\tvm_i\cdot \vu}$, 
it is straightforward to check that the second equality is a direct consequence of the first one 
upon substitution of \eqref{expPhiE-mod-inv-EM}.
Therefore, it is sufficient to prove the relation between $\Msf_n$ and $M_n$.

To prove \eqref{expPhiE-mod-inv3-EM}, we substitute it into the r.h.s. of \eqref{rel-hPhi-tPhi}. 
The resulting expression can be rearranged as
\be 
\sum_{\cI\subseteq \Zv_n} \delta_{\tvu_{\Zv_n\setminus\cI}}\,
\Delta_{n-|\cI|}(\tcM_{\Zv_n\setminus\cI})\,M_{|\cI|}(\cM_\cI;\vu_\cI),
\label{expr-ebM}
\ee
where 
\be 
\Delta_{n}(\cM)=\sum_{\cJ\subseteq\Zv_n}
\esf_{|\cJ|}(\tcM_{\cJ})\,\bsf_{n-|\cJ|+1}(\cM_{\Zv_n\setminus\cJ}).
\label{DelM}
\ee 
It is easy to see that, upon substitution of \eqref{def-bsf-EM} into \eqref{DelM}, the result vanishes unless $n=0$.
Indeed, for all terms with $\cJ\ne\emptyset$, we can set $\cJ=\cI_{m+1}$ which, combined with subsets $\cI_s$, $s=1,\dots,m$,
gives rise to a decomposition of $\Zv_n$.
Then, since $\tcM_{\cJ}=\cM_{\cI_{m+1}\perp \cup_{s=1}^{m}\cI_s}$, the sum of such terms produces
$\bsf_{n+1}(\cM)$, which is canceled by the contribution from $\cJ=\emptyset$. Thus, $\Delta_{n}(\cM)=\delta_n$ and  
the expression \eqref{expr-ebM} reduces to $M_n(\cM;\vu)$, which completes the proof.
\end{proof}

\subsection{A shift}

In fact, we need a slight generalization of the complementary generalized error functions described above.
It takes its origin in the mismatch of the arguments in \eqref{Erefsim} and \eqref{defSlin}: 
while the former involves $\bfq+\beta\bftet$, the latter is evaluated at $\beta=0$ (see \eqref{Efref-lim}).
In other words, if we apply the relation \eqref{expPhiE-mod-EM} (or \eqref{expPhiE-mod3-EM})
between generalized error functions to \eqref{Erefsim}, 
we expect to get sign functions with $\beta$-dependent arguments $\bfv_\ell\ast(\bfq+\beta\bftet)$,
whereas in the sign functions appearing in \eqref{defSlin} the second term is absent.
This suggests that we need another relation which involves sign functions with a shifted argument.
Such a relation can again be obtained by modifying the functions $M_n$ (or $\Msf_n$).

The required modification has been introduced in \cite{Alexandrov:2019rth}
and, away from the discontinuity loci, it is defined as a generalization of \eqref{modMn} depending
on an additional vector of parameters, which appears in the shift of the integration contour,
\be 
\hM_n(\cM;\vu,\vb)=\(\frac{\I}{\pi}\)^n |\det\cM|^{-1} \int\limits_{\IR^n-\I (\vu-\vb)}\de\vu'\,
\frac{e^{-\pi \vu'^2 -2\pi \I \vu'\cdot \vu}}{\prod_{i=1}^n(\cM^{-1}\vu')_i}\, .
\label{mod-hMn}
\ee
Instead of \eqref{expPhiE-mod-EM}, it now satisfies 
\be 
E_n(\cM;\vu)=\sum_{\cI\subseteq \Zv_n}\hM_{|\cI|}(\cM_\cI;\vu_\cI,\vb_\cI)
\prod_{j\in \Zv_n\setminus \cI}\sign(\vm_{j\perp\cI}\cdot(\vu-\vb)),
\label{expPhiE-mod-EhM}
\ee
where $\vb_\cI=\cB_\cI\vb$. This relation can also be taken as a definition of $\hM_n$ extending it to the discontinuity loci.
Similarly, one can define the generalization of the second extension, which we call $\hMsf_n$, by imposing   
\be
E_n(\cM;\vu)=\sum_{\cI\subseteq \Zv_n}\hMsf_{|\cI|}(\cM_\cI;\vu_\cI)\,
\Ssf_{n-|\cI|}(\tcM_{\Zv_n\setminus\cI};(\tvu-\tvb)_{\Zv_n\setminus\cI}).
\label{expPhiE-mod3-hEM}
\ee
All other relations from the previous subsection have their own counterparts as well and are simply obtained by replacing $M_n$ by $\hM_n$,
$\Msf_n$ by $\hMsf_n$, and shifting $\vu$ in the arguments of sign functions and Kronecker symbols by $-\vb$.

\subsection{Boosted functions}
\label{subsec-boost}

As mentioned above, the main use of the generalized error functions is to construct kernels of indefinite theta series.
In such constructions, $\IR^n$, the space of the parameters $\vu$, appears as the negative or positive definite subspace, depending on conventions, 
of a larger space $\IR^d$ endowed with a quadratic form of indefinite signature.
In this situation the parameters of the generalized error functions are defined in terms of objects living in this larger space.
Due to this, it is convenient to introduce {\it boosted} generalized error functions (often, abusing terminology,  
the adjective `boosted' is omitted).
 
To define such functions, let $\{\vbbm_i\}$ be a set
of $n$ vectors of dimension $d$, 
and it is assumed that these vectors span a positive definite subspace in $\IR^d$ endowed with 
a bilinear form
$\ast$, i.e. $\vbbm_i\ast\vbbm_j$ is a positive definite matrix. We also introduce an orthonormal basis 
$\cB=\{\ebbm_i\}$ for this subspace and set
\be 
\cM_{ij}=\ebbm_i\ast\vbbm_j.
\qquad 
u_i=\ebbm_i\ast\xbbm,
\qquad 
b_i=\ebbm_i\ast\bbbm,
\label{M-u}
\ee 
where $\xbbm\in\IR^d$. Then the boosted versions of all generalized error functions considered above
are obtained by evaluating them with the arguments specified by \eqref{M-u}.
For example, we have
\be
\begin{split} 
\Phi_n^E(\{\vbbm_i\};\xbbm)=&\, E_n(\cM;\vu),
\\
\hPhi_n^M(\{\vbbm_i\};\xbbm,\bbbm)=&\,\hM_n(\cM;\vu,\vb),
\\
\hPhi_n^\Msf(\{\vbbm_i\};\xbbm,\bbbm)=&\,\hMsf_n(\cM;\vu,\vb).
\end{split} 
\label{generrPhiME}
\ee

Let us list the most important properties of these functions:
\begin{itemize}
	\item  
	All the boosted functions do not depend on the choice of the basis $\cB$. 
	
	\item 
	$\Phi_n^E$ are smooth functions of $\xbbm$, while $\hPhi_n^M$ and $\hPhi_n^\Msf$ have discontinuities across 
	$\tvbbm_i\ast(\xbbm-\bbbm)=0$, where $\tvbbm_i$ are the dual vectors to $\vbbm_i$. 
	The latter two functions differ only at these discontinuity loci.
	
	\item 
	If $\vbbm_i\ast\vbbm'_j=0$ for $i=1,\dots,n$ and $j=1,\dots,n'$, then 
	\be 
	\Phi_{n+n'}^E(\{\vbbm_i,\vbbm'_j\};\xbbm)=\Phi_n^E(\{\vbbm_i\};\xbbm)\,\Phi_{n'}^E(\{\vbbm'_j\};\xbbm),
	\label{Phi-orth}
	\ee 
	and similarly for other boosted functions. 
	
	\item 
	$\Phi_n^E$ solve the Vign\'eras equation  
	\be
	\[\p_\xbbm^2+2\pi (\xbbm\ast\p_\xbbm-\lambda)\]\Phi(\xbbm)=0
	\label{Vigdif}
	\ee
	with $\lambda=0$, which ensures that the theta series with kernels constructed from $\Phi_n^E$ are modular (or Jacobi) forms.
	Thus, these functions are the key elements to construct modular completions of indefinite theta series.
	
	\item 
	If one of the vectors is null, it reduces the rank of the generalized error function.
	Namely, for $\vbbm_\ell^2=0$, one has
	\be
	\Phi_n^E(\{\vbbm_i\};\xbbm)=\sgn (\vbbm_\ell\,\ast\xbbm)\,\Phi_{n-1}^E(\{\vbbm_i\}_{i\ne \ell};\xbbm).
	\label{Phinull}
	\ee
	In other words, for such vectors the completion is not required.
	
	\item 
	The relations between various types of boosted functions are obtained by substituting \eqref{M-u} into the relations 
	presented in the previous subsections.
	Here we provide only those that are used in the main text:
	\be
	\Phi_n^E(\{\vbbm_i\};\xbbm)=\sum_{\cI\subseteq \Zv_n}\hPhi_{|\cI|}^M(\vbbm_\cI;\xbbm,\bbbm)
	\prod_{i\in \Zv_n\setminus \cI}\sign(\vbbm_{i\perp \cI}\ast(\xbbm-\bbbm)),
	\label{expPhiE-mod}
	\ee
	\be
	\hPhi_n^\Msf(\{\vbbm_i\};\xbbm,\bbbm) =
	\sum_{\cI\subseteq \Zv_{n}}\biggl(\bsf_{n-|\cI|+1}(\vbbm_{(\Zv_n\setminus\cI)\perp\cI})
	\prod_{i\in \Zv_n\setminus \cI}  \delta_{\vbbm_{i\perp \cI}\ast(\xbbm-\bbbm)}\biggr)\,
	\hPhi^M_{|\cI|}(\vbbm_\cI;\xbbm,\bbbm),
	\label{expPhiE-mod-inv3}
	\ee
	where $\vbbm_\cI=\{\vbbm_i\}_{i\in\cI}$, $\vbbm_{\cI\perp\cJ}=\{\vbbm_{i\perp\cJ}\}_{i\in\cI}$, 
	and\footnote{We slightly abused the notation and denoted $\esf_n(\{\vbbm_i\})=\Phi^E_n(\{\vbbm_i\};0)$, 
		which is the same as $\esf_n(\cM)$ with $\cM$ from \eqref{M-u}.} 
	\be  
	\bsf_{n+1}(\{\vbbm_i\})=\sum_{m=1}^n(-1)^m \sum_{\cup_{s=1}^m \cI_s =\Zv_n}\prod_{s=1}^m 
	\esf_{|\cI_s|}\Bigl(\vbbm_{\cI_s\perp \cup_{\ell=1}^{s-1}\cI_\ell}\Bigr).
	\label{def-bsf}
	\ee 
\end{itemize}

\section{Numerical coefficients}
\label{ap-coef}

In the main text, we extensively use the two sets of coefficients, $e_n$ and $b_n$.
They are defined in \eqref{def-genthm} and \eqref{defbn} and coincide with the Fourier coefficients 
of $\mbox{Arctanh}(x)/x$ and $\tanh(x)$, respectively.
Since the two functions (up to $1/x$ factor) are inverse of each other, both sets of coefficients are mutually related.
The easiest way to derive such a relation is to substitute the Taylor expansions into the identity
$\tanh(\mbox{Arctanh}(x))=x$. Collecting terms with equal powers of $x$, one obtains 
\be
\sum_{m= 1}^{n}b_m\sum_{\sum_{s=1}^{m}n_s=n}
\prod_{s=1}^{m} e_{n_s-1}=\delta_{n-1}.
\label{ident-be}
\ee 
It is obvious that in the same way the inverse identity $\mbox{Arctanh}(\tanh(x))=x$ leads to a similar relation with the roles of 
$b_m$ and $e_{m-1}$ interchanged.
For $n>1$, this relation can also be written as a recursive relation for $b_n$'s
\be
b_n=-\sum_{m=3}^n\sum_{\sum_{s=1}^m n_s=n}e_{m-1} \prod_{s=1}^m b_{n_s}
=-\sum_{T\in \IT^{\rm S,\leq 2}_n} 
e_{n_{v_0}-1}\prod_{v\in V_T\setminus\{v_0\}}b_{n_v},
\label{b-eb}
\ee 
where in the second equality we used the equivalence between sums over ordered decompositions and Schr\"oder trees of depth at most 2
discussed in \S\ref{subsec-useful}. Applying it recursively, one arrives at an explicit expression of $b_n$ in terms of $e_n$'s:
\be
b_n=\sum_{T\in \IT_{n}^{\rm S}}(-1)^{n_T} \prod_{v\in V_T} e_{n_{v}-1} , 
\label{expr-bn}
\ee 
where now the sum goes over Schr\"oder trees without any restriction on their depth. 

For convenience, we provide here the first few non-vanishing coefficients for both sets:
\be
\begin{array}{rlrlrlrlrlrl}
	e_0&=1, &\quad e_2&=\frac13\, , &\quad e_4&=\frac15\, , &\quad  e_6&=\frac17\, , &\quad e_8&=\frac19\, , &\quad e_{10}&=\frac{1}{11}\, ,
	\\
	b_1&=1, &\quad b_3&=-\frac13\, , &\quad b_5&=\frac{2}{15}\, , &\quad  b_7&=-\frac{17}{315}\, , &\quad b_9&=\frac{62}{2835}\, ,
	&\quad b_{11}&=-\frac{1382}{155925}.
\end{array}
\ee 

\medskip 

Both sets of coefficients turn out to be related to more general coefficients appearing in appendix \ref{ap-generr}, 
$\esf_n(\{\vbbm_i\})$ and $\bsf_n(\{\vbbm_i\})$ defined in \eqref{def-bsf}.
The relation is described by the following equations
\begin{subequations}
	\bea 
	\cSym\esf_n(\{\bfv_i\})&=&e_n,
	\label{esym}
	\\
	\Sym\bsf_n(\{\bfv_i\})&=&b_n,
	\label{bsym}
	\eea
\label{rel-eb-ebsf}
\end{subequations}
where $\cSym$ denotes the sum over cyclic permutations of charges weighted by $1/n$.
The first relation was demonstrated in \cite[eq.(2.27)]{Alexandrov:2024jnu}.
To prove the second relation, we apply Lemma \ref{lemma1} to the definition \eqref{def-bsf}.
Then following the same reasoning as in \S\ref{subsec-useful}, one can rewrite it as\footnote{In fact, the relation between 
the ordered decompositions into subsets in \eqref{def-bsf} and the Schr\"oder trees is not one-to-one.
There are decompositions whose contributions differ only by sign and cancel each other. We have checked \eqref{def-bsf2} on many examples.} 
\be  
\bsf_n(\{\bfv_i\})
=\sum_{T\in \IT_{n}^{\rm S}}(-1)^{n_T} \prod_{v\in V_T} \esf_{n_{v}-1} (\bfv_{\Ch(v)}),
\label{def-bsf2}
\ee 
where $\bfv_{\Ch(v)}$ is the set of vectors $\bfv_\ell$ constructed from the charges of children of $v$.
Applying symmetrization to this expression, one can make it to act on each vertex separately. 
This allows to apply \eqref{esym}, which replaces all $\esf_{n_{v}-1}$ by $e_{n_{v}-1}$.
As a result, one obtains the same expression as on the r.h.s. of \eqref{expr-bn}
and hence it is equal to $b_n$.

\section{Solution of $\Gamma_\ell(J)=0$}
\label{ap-lat}

In this appendix we solve equations 
\be 
\Gamma_\ell(J)=0, 
\qquad 
\ell=1,\dots, n-1,
\label{eq-Gl}
\ee 
which are equivalent to 
\be 
J\cdot q_i=\frac{N_i}{N}\, J\cdot \mu,
\qquad
i=1,\dots, n.
\label{eq-Jqi}
\ee 

The cleanest way to solve these equations is to use the technique of lattice factorization (see, e.g., \cite{CSbook}).
The point is that the conditions \eqref{eq-Jqi} suggest to consider a decomposition of the lattice $\bfLam^{(\bfN)}_{\IF_m}$ \eqref{latVW}
into two orthogonal sublattices with vectors proportional to $J$ and $J_\perp=\veps\cdot J$
(see above \eqref{prop-JJp} for the definition of matrix $\veps$).
However, the sum of the two sublattices does {\it not} coincide with the full lattice.
It fails to account for the elements that are known as {\it glue vectors}. 

As a simple example, let us consider the decomposition of $\Lambda_{\IF_m}$ induced by $J$,
which we assume to have $\gcd\{J^\alpha\}=1$. In this case it is given by
\be 
\Lambda_{\IF_m}=\bigcup_{a=0}^{J^2-1}\Bigl(J\IZ\oplus J_\perp\IZ+\glu_a\Bigr),
\label{fact-JLambda}
\ee 
where $\glu_a\in\Lambda_{\IF_m}$ are the glue vectors which are labeled by the index $a=0,\dots, J^2-1$.
It is easy to see that the glue vectors can be chosen to be
\be 
\glu_a=a\vrh, 
\ee
where $\vrh\in\Lambda_{\IF_m}$ is any fixed vector such that $J\cdot \vrh=1$.

It is straightforward to generalize the decomposition \eqref{fact-JLambda} to the lattice $\bfLam^{(\bfN)}_{\IF_m}$
and to its dual. However, there is a small complication in our case due to the fact that $\bfq$ is not really an element
of the dual lattice, but differs from it by a shift by $\frac{\mu}{N} \boldsymbol{1}$.
Taking this into account, one can show that the most general charge vector can be represented as 
\be 
q_i=N_i\(J k^\parallel_i+J_\perp k^\perp_i +a_i\vrh\)+\mu_i\, ,
\label{genqi}
\ee 
where $k^\parallel_i,k^\perp_i \in\IZ$, $a_i\in\IZ_{J^2}$ and they satisfy the following constraints
\be
\sum_{i=1}^n  N_i k^\parallel_i =0 ,
\qquad 
\sum_{i=1}^n  N_i k^\perp_i =\vrh\cdot\veps\cdot\Delta\mu,
\qquad
\sum_{i=1}^n  a_i=J\cdot \Delta\mu,
\label{constr-sum}
\ee 
with $\Delta\mu=\mu-\sum_{i=1}^n\mu_i$.
As a crosscheck, one can verify 
\be 
\sum_{i=1}^nq_i=J_\perp(\vrh\cdot\veps\cdot\Delta\mu)+\vrh(J\cdot \Delta\mu)+\sum_{i=1}^n\mu_i=\mu,
\ee
where the second equality holds due to the identity
\be
\veps^{\alpha\beta}\veps_{\gamma\delta}=\delta^\alpha_\delta\delta^\beta_\gamma-\delta^\alpha_\gamma\delta^\beta_\delta.
\label{ident-veps}
\ee

Imposing \eqref{eq-Jqi}, one finds that
\be
a_i=J\cdot \(\frac{\mu}{N}-\frac{\mu_i}{N_i}\) \mod J^2,
\ee 
which satisfies the constraint in \eqref{constr-sum}.
Requiring the r.h.s. to be integer, we reproduce the constraints \eqref{cond-mumu} on the residue classes.
Substituting $a_i$ back into \eqref{genqi}, taking into account that, up to the diagonal term proportional to $\mu$, 
only the projection on $J_\perp$ is non-trivial, and using again \eqref{ident-veps}, one obtains
\be 
q_i=J_\perp\(N_ik^\perp_i+\vrh\cdot\veps\cdot\mu_i\)+\frac{N_i}{N}\, (J\cdot\mu)\vrh.
\label{sol-G0}
\ee 

\providecommand{\href}[2]{#2}\begingroup\raggedright\endgroup


\end{document}